\documentclass[11pt,a4paper]{amsart}
\usepackage{amsthm}
\usepackage{amsmath}
\usepackage{amscd}
\usepackage[latin2]{inputenc}
\usepackage{t1enc}
\usepackage[mathscr]{eucal}
\usepackage{indentfirst}
\usepackage{graphicx}
\usepackage{graphics}
\usepackage{pict2e}
\usepackage{epic}
\numberwithin{equation}{section}
\usepackage[margin=3cm]{geometry}
\usepackage{epstopdf} 
\usepackage{bbold}
\usepackage{amsfonts}
\usepackage{amssymb}
 \usepackage[foot]{amsaddr}
\usepackage{mdframed}
\usepackage{hyperref}

 \usepackage{tikz}

\usepackage[n, operators, sets, adversary, landau, probability, notions, logic, ff, mm, primitives, events, complexity, oracles, asymptotics, keys]{cryptocode}

\theoremstyle{plain}
\newtheorem{theorem}[equation]{Theorem}
\newtheorem{Lemma}[equation]{Lemma}
\newtheorem{Cor}[equation]{Corollary}
\newtheorem{Prop}[equation]{Proposition}

 \theoremstyle{definition}
\newtheorem{Def}[equation]{Definition}
\newtheorem{fact}[equation]{Fact}

\newtheorem{Rem}[equation]{Remark}
\newtheorem{?}[equation]{Problem}
\newtheorem{Ex}[equation]{Example}

\newcommand{\pcce}{\mathtt{pcce}}
\newcommand{\hwce}{\mathtt{hwce}}
\newcommand{\wce}{\mathtt{wce}}
\newcommand{\quadr}{\mathtt{quadr}}
\newcommand{\mmq}{\mathsf{m}}
\newcommand{\nmq}{\mathsf{n}}
\newcommand{\bintoint}{\mathtt{int}}
\newcommand{\increase}{\mathtt{increase}}
\renewcommand{\abs}[1]{\lvert#1\rvert}

\begin{document}

\title{On the equivalence of two post-quantum cryptographic families}

\author[A. Meneghetti]{Alessio Meneghetti}

\address[A. Meneghetti]{University of Trento, Department of Mathematics, \texttt{\textup{alessio.meneghetti@unitn.it}}}

\author[A. Pellegrini]{Alex Pellegrini}

\address[A. Pellegrini]{Eindhoven University of Technology, Department of Mathematics and Computer Science, \texttt{\textup{a.pellegrini@tue.nl}}}

\author[M. Sala]{Massimiliano Sala}

\address[M. Sala]{University of Trento, Department of Mathematics, \texttt{\textup{massimiliano.sala@unitn.it}}}

\begin{abstract}
The Maximum Likelihood Decoding Problem (MLD) is known to be NP-hard and its complexity is strictly related to the security of some post-quantum cryptosystems, that is, the so-called code-based primitives.
Analogously, the Multivariate Quadratic System Problem (MQ) is NP-hard and its complexity is necessary for the security of the so-called multivariate-based primitives.
In this paper we present a closed formula for a polynomial-time reduction from any instance of MLD to an instance of MQ, and viceversa. We also show a polynomial-time isomorphism between MQ and MLD, thus demonstrating the direct link between the two post-quantum cryptographic families.
\end{abstract}

\maketitle





\section{Introduction} 
Computationally difficult algebraic problems, e.g. the discrete logarithm problem (DLP) and the integer factorisation problem (IP),  have historically been successfully exploited to construct secure cryptographic protocols, such as e.g. RSA, Diffie-Hellmann and ECDSA. An interested reader can find details on these schemes and the underlying algebraic problems in \cite{menezes2018handbook} or \cite{stinson2019cryptography}. The advent of quantum algorithms, for example Shor's algorithm \cite{shor1994algorithms}, has however incited several established primitives vulnerable to quantum attackers, thus encouraging researchers to design and analyse new families of cryptosystems based on NP-hard problems, since no quantum algorithm is known to be able to efficiently solve them. Post-Quantum Cryptography \cite{bernstein2009} refers to the the class of cryptographic primitives built upon computational problems not readily solvable by quantum computers. Among these, of particular interest we list lattice-based, code-based and multivariate-based cryptosystems, namely, systems whose security is intertwined with the computational complexity of solving problems over lattices, e.g. the shortest vector problem (SVP) or the closest vector problem (CVP) \cite{micciancio2002complexity}, problems based on coding theory, e.g. the Maximum Likelihood Decoding Problem (MLD) \cite{Berlekamp}, and problems over polynomial ideals, e.g. the problem of deciding whether a quadratic Boolean polynomial system admits a solution, usually referred to as Multivariate Quadratic Problem (MQ) \cite{garey}.
\\
This massive research effort is still ongoing, as can be noticed by the wide participation to the NIST standardisation process for post-quantum primitives\footnote{\url{https://csrc.nist.gov/projects/post-quantum-cryptography}, accessed on 2022-01-18.}. The round-3 finalists for key-encapsulation mechanisms (KEM) are Classic McEliece \cite{bernstein2017classic}, CRYSTALS-KYBER \cite{avanzi2017crystals}, NTRU \cite{chen2019NTRU} and SABER \cite{bassoSABER}, while the round-3 finalists for digital signature schemes are CRYSTALS-DILITHIUM \cite{ducas2021crystals}, FALCON \cite{fouque2018falcon} and Rainbow \cite{Rainbow}. Notably, the only round-3 finalists which are not lattice-based are Classic McEliece and Rainbow, which are respectively code-based and multivariate-based. 

\medskip

As briefly stated above, code-based schemes are designed by exploiting computational and decision problems obtained by questions arising from Coding Theory, such as for example the Maximum-Likelihood Decoding problem (MLD) for linear codes. Using the words of Guruswami and Vardy, ``MLD is one of the central (perhaps, the central) algorithmic problems in coding theory" \cite{Guruswami}. 
This problem, given here in Definition \ref{def:MLD}, has been proven to be NP-complete in 1978 by Berlekamp, McEliece and van Tilborg in \cite{Berlekamp}, where the authors provide a reduction from the Three-Dimensional Matching problem for graphs \cite{Karp}. Then, in 1990, MLD was proven to be in P/Poly by Bruck and Naor \cite{Bruck} and by Lobstein \cite{Lobstein}.
Other interesting results on the complexity of MLD were then given by several authors, and an interested reader can find more information in \cite{Berlekamp, Vardy} for the general case and in \cite{Cheng,Gandikota, Guruswami, Peterson} for specific classes of codes.
\\
In the context of post-quantum code-based cryptography, the most famous cryptosystem is that proposed by McEliece in 1978 \cite{mceliece1978public}. This scheme has been studied for over forty years, proving its resilience and security, and it was then used, together with the Niederreiter scheme \cite{niederreiter1986knapsack}, as a building block for the NIST round-3 finalist Classic McEliece. We remark that the security of each of these primitives is related to both MLD and the problem of distinguishing between apparently-random codes and permuted versions of algebraic codes. In this work we focus on the first of these two problems. A presentation of code-based cryptography and the underlying problems can be found in the chapter \textit{Code-based cryptography} by Overbeck and Sendrier in \cite{bernstein2009}.
\\
Before stating MLD, we provide some basics on binary linear codes. 
An $[n,k]$ linear code $C$ is a $k$ dimensional subspace of $\left( \mathbb{F}_q \right)^n$. The parameters $n$ and $k$ are called the \textit{length} and \textit{dimension} of the code. Since we are interested in binary codes, throughout this paper we only consider the case $q=2$, hence with $\mathbb{F}$ we mean $\mathbb{F}_2$ and we often omit the term binary whenever we use the term code. A \textit{codeword} is any vector in the code. A \textit{generator matrix} $G$ of $C$ is a $k\times n$ matrix whose rows span $C$. Similarly, a \textit{parity-check matrix} of $C$ is a generator matrix of the dual of $C$. Due to this definition, a vector $c$ is a codeword if and only if $H\cdot c^{\top} = 0$. We recall that if $G$ is systematic, i.e. $G=\begin{bmatrix}I_k\mid R\end{bmatrix}$, then $H=\begin{bmatrix}-R^{T}\mid I_{n-k}\end{bmatrix}$.
The (Hamming) \textit{weight} of a vector $v$ is the number $\mathrm{w}(v)$ of its non-zero components.
\begin{Def}[MLD]\label{def:MLD}
Let $H=[h_{i,j}]_{i=1,\ldots m,j=1,\ldots,n}$ be an $m\times n$ binary matrix, let $s\in\mathbb{F}^m$ and let $t\leq n$ be a positive integer. Decide whether there is a vector $v \in \mathbb{F}^n$ of weight at most $t$, such that $Hv^\top = s^\top$.
\end{Def}
\noindent We denote with $I_{\mathrm{MLD}}$ a generic instance of MLD, determined by the triple $I_{\mathrm{MLD}} = (H,s,t)$. The values $(n,m)$ determine the memory space required to store $I_{\mathrm{MLD}}$, indeed, we need a total of $nm+m+\left\lfloor\log_2n\right\rfloor+1$ bits to write a given instance $(H,s,t)$. We call $(n,m)$ the complexity parameters of $I_{\mathrm{MLD}}$, and $\abs{I_{\mathrm{MLD}}}=nm+m+\left\lfloor\log_2n\right\rfloor+1$ the size of $I_{\mathrm{MLD}}$.
We can assume that $m\leq n$, since this case has the same hardness as the general case.

\medskip

The second problem we consider is to decide whether a multivariate-quadratic Boolean system admits a solution. 
This problem, known as the ``multivariate quadratic equation system problem" (MQ) is linked to the security of multivariate-based cryptosystems (e.g. Oil\&Vinegar \cite{kpg99}, Rainbow \cite{Rainbow}, GeMSS \cite{GeMSS}). \\
Let $I$ be an ideal of a polynomial ring $\mathcal{R}$ over a field $\mathbb{K}$ and let $\mathbb{E}$ be an extension field of $\mathbb{K}$, we denote by 
$$\mathcal{V}_{\mathbb{E}}(I)= \left\{A \in \mathbb{E}^{\nmq} \mid f(A) = 0 \; \forall f \in I\right\}$$ 
the set of all the zeroes of $I$ in $\mathbb{E}^{\nmq}$.
$\mathcal{V}_{\mathbb{E}}(I)$ is called the \textit{variety} of $I$ over $\mathbb{E}$.
An MQ-system of equations over $\mathbb{F}$ is a set of $\mmq$ polynomial equations of degree at most 2 in $\mathbb{F}[x_1,\ldots,x_{\nmq}]$ of the form:
\begin{equation}\label{eq:polysystem}
S = \left \{
	\begin{array}{c}
		f_1(x_1,\ldots,x_{\nmq}) = 0 \\
		f_2(x_1,\ldots,x_{\nmq}) = 0 \\
		\vdots\\
		f_{\mmq}(x_1,\ldots,x_{\nmq}) = 0 
	\end{array}
	\right .
\end{equation}
where for every $h \in \{1,\ldots,\mmq\}$
\begin{equation}\label{eq: MQ instance}
	f_h(x_1,\ldots,x_{\nmq}) = \sum_{1\le i < j \le \nmq}\gamma_{ij}^{(h)}x_ix_j + \sum_{1\le i \le \nmq}\lambda_i^{(h)}x_i + \delta^{(h)}
\end{equation}
with $\gamma_{ij}^{(h)}, \lambda_i^{(h)}, \delta^{(h)} \in \mathbb{F}$.\\
The \textit{(decision) multivariate quadratic equation system problem} (MQ) can now be stated as
\begin{Def}[MQ]\label{def:mq}
Consider a polynomial system $S = \{f_1,\ldots,f_{\mmq}\}$ as in (\ref{eq:polysystem}) of degree at most 2 over $\mathbb{F}$ and let $I$ be the ideal generated by $S$. \\
Decide whether $\mathcal{V}_{\mathbb{F}}(I)$ is non-empty.
\end{Def}
\noindent We denote with $I_{\mathrm{MQ}}$ a generic instance of MQ, determined by the polynomial system $S$.
Similarly to the case of MLD, the values $(\nmq,\mmq)$ determine the memory space required to store an instance $I_{\mathrm{MQ}}$. In this case, we need a total of $\mmq\left(\binom{\nmq}{2}+\nmq+1\right)$ bits to write $S$. We call $(\nmq,\mmq)$ the complexity parameters of $I_{\mathrm{MQ}}$, and $\abs{I_{\mathrm{MQ}}}=\mmq\left(\binom{\nmq}{2}+\nmq+1\right)$ its size.

MQ has been proven to be NP-hard over any field \cite{garey}, and many cryptosystems rely their security on such problem \cite{kpg99,mi88,pat96}.
Several mathematical approaches have been employed to tackle this problem, such as the Newton and the tensor-based algorithms \cite{dennis,schnabel}, Gr{\"o}bner bases, resultants and eigenvalues/eigenvectors of companion matrices \cite{cox}, semidefinite relaxations \cite{bucero,lasserre,parrillo}, numerical homotopy \cite{li,verschelde}, low-rank matrix recovery \cite{davenport}, and symbolic computation \cite{grigoriev}.
 
The most established method to perform cryptanalysis of public-key systems is to focus on the algebraic problems underlying them. In this work we look at it from a slightly different perspective, establishing a link between MLD and MQ, and thus providing new directions in the analysis of both code-based and multivariate-based primitives.
More precisely, the aim of this paper is to show explicit reductions between the two previous problems. Since both are NP-complete problems, one might be reduced to the other, but it is not obvious how to do it explicitly without losing their algebraic nature.\\
MQ and MLD are problems of a purely algebraic nature, naturally stated and studied in the context of vector spaces and polynomial rings over $\mathbb{F}_2$, the smallest possible field.
Interestingly, their complexity and the complexity of the numerous related problems (including search problems) is at the heart of research for the mathematical community working in coding theory and cryptography. Yet, known results about their complexity are obtained via techniques of a rather different nature, such as graph theory.
As far as we know, this is the first paper that investigates their direct explicit complexity links, using only languages and tools familiar to standard research in coding theory and cryptography alike.
To be more precise, we will establish in Section \ref{sec:MLD-MQ} an explicit reduction from MLD to MQ, while in Section \ref{sec:MQ-MLD} we will present a reduction from MQ to MLD.
The remainder of this paper contains Section \ref{sec: preliminaries}, where we provide preliminaries and our notation, Section \ref{sec: open problems}, where
we leave some open problems, and notably Section \ref{sec: conclusion}. In this latter section we draw some significant conclusions, among which we report our proof of the existence of a polynomial-time isomorphism between NP and MQ, and thus between code-based and multivariate-based primitives.

\section{Preliminary results and definitions}\label{sec: preliminaries}
Before introducing the reductions between MLD and MQ we need some preliminary notation.
Throughout this paper we consider vectors to be row vectors, unless otherwise specified. Moreover, we denote with $\bar{\cdot}$ each element of a set which is not a variable (regardless it being an element of fields or vector spaces), whereas without the notation $\bar{\cdot}$ we mean variables. As an example, if $f(x)$ is a polynomial in the variable $x$ then $f(\bar{x})$ is to be considered as the evaluation of $f$ at the point $\bar{x}$. In Section \ref{sec:MQ-MLD} we will also use the notation $\widehat{\cdot}$ and $\widetilde{\cdot}$ instead of $\bar{\cdot}$ to distinguish between elements belonging to distinct spaces. As an example, in Section \ref{sec:MQ-MLD} we will define two distinct sets $\widehat{\Sigma}$ and $\widetilde{\Sigma}$, which are somewhat linked. In this case, to distinguish between their elements, we will use $\widehat{v}\in\widehat{\Sigma}$ and $\widetilde{v}\in\widetilde{\Sigma}$.
\\
Let $l$ be a positive integer.
We define the map $\bintoint:\mathbb{F}^l \rightarrow \mathbb{Z}$ as
\begin{equation}\nonumber
	\bintoint(a)=\bintoint\left(\left(a_1,\ldots,a_l\right)\right):=\sum_{j=1}^l{a_j \cdot 2^{j-1}}\; ,
\end{equation}
where the sum on the right-hand side is over the integers.
In this way, $\bintoint(\bar{a})$ is the integer value corresponding to the input vector of bits $\bar{a}=(\bar{a}_1,\ldots,\bar{a}_l)\in\mathbb{F}^l$.
When we regard a vector in $\mathbb{F}^l$ as the binary representation of an integer, we list its bits from the least-significant to the most-significant, e.g. if $l=4$ then $3=(1,1,0,0)$. From now on we will often use the binary representation of the parameter $t$ of an MLD instance $(H,s,t)$ as a vector of $\ell=\lfloor\log_2n\rfloor+1$ bits $(t_1,\ldots,t_{\ell})$, and therefore we will use both $\mathrm{w}(v)\leq t$ and $\mathrm{w}(v)\leq \bintoint(t_1,\ldots,t_{\ell})$.
\\
Let $\textbf{1}$ be the vector $(1,\ldots,1)$. 
We define the map $\increase:\mathbb{F}^l \rightarrow \mathbb{F}^l$ such that, for all $\bar{a}\neq \textbf{1}$, we have
\begin{equation}\nonumber
\bintoint(\increase(\bar{a})) = \bintoint(\bar{a}) +1\;,
\end{equation}
namely, if $\bar{a}=\left(\bar{a}_1,\ldots, \bar{a}_z,\ldots,\bar{a}_l\right)$, where $\bar{a}_z=0$ is the left-most 0 bit of $\bar{a}$, then 
\begin{equation}\nonumber
	\increase(\bar{a}) = \begin{cases}
	       \begin{array}{ll}
	           \left(0,\ldots,0,1,\bar{a}_{z+1},\ldots,\bar{a}_l\right)  & \text{ if } \bar{a} \neq \textbf{1} \\
	           \bar{a}  & \text{ if } \bar{a} = \textbf{1}\;.
	       \end{array}

	\end{cases}
\end{equation}

Finally, we introduce the projection and truncation maps $\pi$ and $\tau$. Let $i\leq l$ be a non-negative integer, and let $\pi_i:\mathbb{F}^l \rightarrow \mathbb{F}^l$
the projection defined as
$$\pi_i\left(\bar{v}_1,\ldots,\bar{v}_l \right) = \left(\bar{v}_1,\ldots,\bar{v}_i,0,\ldots,0 \right) $$
for $1\leq i\leq l$ and
$$
\pi_0\left(\bar{v}_1,\ldots,\bar{v}_l \right) = \left(0,\ldots,0 \right)\;.
$$
Similarly, let $i\leq l$ be a positive integer. We define the truncation $\tau_i:\mathbb{F}^l \rightarrow  \mathbb{F}^i$
as
$$\tau_i\left(\bar{v}_1,\ldots,\bar{v}_l \right) = \left(\bar{v}_1,\ldots,\bar{v}_i\right)\;. $$
Consider now a polynomial equation $f=0$, where $f \in \mathbb{F}[x_1,\ldots,x_l]$ with $\mathrm{deg}(f) = d >2$. The goal of the algorithm described in the following fact is to reduce this equation to a set of equations of degree at most $2$. We use an idea similar to that described by Kipnis and Shamir \cite{kipnis} in their \textit{relinearization} technique. 

\begin{fact}
\label{rem:quadratize}
Let $x_{i_1}x_{i_2}\cdots x_{i_d}$ be a monomial with degree $d$. We introduce a set of new $d-2$ variables as follows
\begin{equation}\nonumber
\left\{\begin{array}{ll}
y_1=x_{i_1}x_{i_2}\\
y_2=y_1x_{i_3}\\
\vdots \\
y_{d-2}=y_{d-3}x_{i_{d-1}}
\end{array}
\right.
\end{equation}
and thus rewrite $x_{i_1}x_{i_2}\cdots x_{i_d}$ as $y_{d-2}x_{i_d}$. With this procedure, a monomial of degree $d$ is substituted by a set of $d-1$ quadratic equations by introducing  $d-2$ variables.

By applying the same argument to each monomial of $f$, we obtain a system of quadratic equations, as required.
\end{fact}
For convenience, given a polynomial equation $f=0$ we denote with $\quadr(f)$ the quadratic polynomial system obtained by applying the procedure in Fact \ref{rem:quadratize} to $f$. Similarly, given a polynomial system $S$, we denote with $\quadr(S)$ the quadratic system obtained by applying the procedure to each of the polynomials in $S$ and joining all the systems of equations, formally
\begin{equation}
    \quadr(S) = \bigcup_{f\in S}\quadr(f).
\end{equation}
\begin{Ex}
    Let $S=\{f_1,f_2\}$ be the polynomial system
    $$
    S=
    \left\{
    \begin{array}{l}
    f_1= x_1x_2x_4+x_1x_3x_4+x_2x_3+x_1=0\\
    f_2=x_1x_2x_3x_4+1=0
    \end{array}
    \right.
    $$
    To compute $\quadr(S)$ we start with computing $\quadr(f_1)$. There are two monomials of degree larger than $2$ in $f_1$, namely $x_1x_2x_4$ and $x_1x_3x_4$, which have degree $d_1=3$. When dealing with a monomial with degree $d_1=3$, as stated in Fact \ref{rem:quadratize}, we introduce $d_1-2=1$ variable and then we obtain a set of $d_1-1=2$ quadratic equations (including the rewriting of the starting equation $f_1=0$ in terms of the new variables).
    \\
    We start from $x_1x_2x_4$, we call the new variable $y_1$, and we create the quadratic equation $y_1=x_1x_2$. Due to this new equation, we can substitute $x_1x_2$ with $y_1$ into $f_1=0$, obtaining
    $$
    f_1=0 \quad\Leftrightarrow\quad
    \left\{
    \begin{array}{l}
    y_1=x_1x_2\\
    y_1x_4+x_1x_3x_4+x_2x_3+x_1=0
    \end{array}
    \right.
    $$
    Similarly, we deal with the monomial $x_1x_3x_4$ by introducing $y_2=x_1x_3$. The system becomes
    $$
    \quadr(f_1)=\left\{
    \begin{array}{l}
    y_1=x_1x_2\\
    y_2=x_1x_3\\
    y_1x_4+y_2x_4+x_2x_3+x_1=0
    \end{array}
    \right.
    $$
    We proceed now with the second equation. In $f_2$ only $x_1x_2x_3x_4$ has degree larger than $2$, and in this case we need to introduce two variables $z_1,z_2$ and transform $f_2=0$ into a system of $3$ quadratic equations:
    $$
    \quadr(f_1)=
    \left\{
    \begin{array}{l}
    z_1=x_1x_2\\
    z_2=z_1x_3
    z_2x_4+1=0
    \end{array}
    \right.
    $$
    Thus $\quadr(S) = \quadr(f_1) \cup \quadr(f_2)$ is a  quadratic system equivalent to $S$.
\end{Ex}

\begin{Def}\label{def:standardform}
A system of equations is said to be in \textit{standard form} if
\begin{itemize}
    \item it contains equations of the form $xy + z =0$ which do not share any variable;
    \item it contains linear equations with up to three monomials, that is, of the form $x+\bar{\delta}=0$, $x+y+\bar{\delta}=0$ or $x+y+z+\bar{\delta}=0$, with $\bar{\delta}\in\mathbb{F}$;
    \item each variable that appears in a linear equation appears also in exactly one quadratic equation;
    \item it does not contain any other kind of equation.
\end{itemize}
\end{Def}

In Lemma \ref{lem:systostandard} we will show that any system of quadratic equations can be brought to standard form by adding (a bounded number of) new variables and equations. A preliminary result for the case of linear equations is the following:
\begin{fact}\label{fact: linear standard}
Let us consider a linear equation with $l$ variables
$$
x_1+x_2+x_3+\ldots+x_l=0 \;.
$$
If we define $y_i$ to be the sum of the the first $i$ variables $x_1+\ldots+ x_i$, then $x_1+x_2+x_3+\ldots+x_l=0$ is equivalent to the linear system
$$
\left\{
\begin{array}{l}
x_1+x_2+y_2=0  \\
y_2+x_3+y_3=0  \\
y_3+x_4+y_4=0\\
\vdots\\
y_{l-3}+x_{l-2}+y_{l-2}=0\\
y_{l-2}+x_{l-1}+x_{l}=0
\end{array}
\right.
$$
which has $l-2$ equations, each one involving exactly three variables, and a total of $2l-3$ variables.
\end{fact}

\subsection{Our Computational Model}
To establish reductions' performances between MLD and MQ we need a computational model that defines the computational cost of operations performed in the reduction. In our case, the set of operations we need in order to perform a reduction are the sum and multiplication in the finite field $\mathbb{F}$. The two operations can be identified with the $\mathrm{OR}$ and $\mathrm{AND}$ logical operators, respectively, to which we assign cost $1$. We also need to carefully consider the memory requirements of our method, so we assign cost $1$ to every coefficient required to express a single polynomial. Notice that the number of bits required to completely define a single quadratic polynomial is approximately the square of the number of variables.
\\
In Section \ref{sec:MLD-MQ} we will present a reduction $\alpha: \mathrm{MLD}\to \mathrm{MQ}$, whose memory analysis will be based on the number of equations and variables required to describe the related polynomials (i.e. the complexity parameters of the resulting $\mathrm{MQ}$ instance).
In Section \ref{sec:MQ-MLD} we present a reduction $\beta: \mathrm{MQ}\to \mathrm{MLD}$ whose memory analysis is based on the  size of the generated parity-check matrix in the reduction (i.e. the complexity parameters of the resulting $\mathrm{MLD}$ instance).
\\
The following lemma, together with Fact \ref{fact: linear standard}, implies that every instance of MQ can  be reduced to an instance in which the polynomial system is given in standard form.

\begin{Lemma}\label{lem:systostandard}
	Consider a polynomial system $S = \lbrace f_1,\ldots,f_{\mmq} \rbrace$ with $f_i \in \mathbb{F}[x_1,\ldots,x_{\nmq}]$ and $\mathrm{deg}(f_i)=2$ for each $i=1,\ldots,\mmq$. $S$ can be taken to standard form in $\mathcal{O}(\mmq\nmq^2)$ operations.\\
	More precisely, the number of quadratic equations is at most 
$\mmq\left(\frac{\nmq(\nmq-1)}{2}\right)$
	and the number of linear equations is at most
	$\mmq\left(\frac{3\nmq^2-\nmq}{2}-2\right)$.
\end{Lemma}

\begin{proof}
Assume first $S=\lbrace f \rbrace$. In the worst case 
\begin{equation}\nonumber
f = \sum_{i< j \in\{1,\ldots,\nmq\}}x_ix_j + \sum_{i\in\{1,\ldots,\nmq\}}x_i + \bar{\delta}\;,
\end{equation}
with $\bar{\delta}\in\mathbb{F}$. Clearly, $f$ has $\binom{\nmq}{2}=\frac{\nmq(\nmq-1)}{2}$ quadratic monomials $x_ix_j$. As a first step, we introduce $\binom{\nmq}{2}$ new variables, along with a set of $\binom{\nmq}{2}$ equations of the form $x_{ij} +x_ix_j = 0$. These newly introduced variables are then substituted into $f$, obtaining in this way a linear polynomial $f'$ with $\frac{\nmq(\nmq+1)}{2}$ variables.
\\
Notice that the quadratic equations we just introduced share some variables $x_i$: for each $i$, $x_i$ appears indeed in $\nmq-1$ such equations. However, our aim is to have variables in degree-2 monomials appearing in exactly one monomial. This is achieved in a second step by introducing new variables and new linear equations: if $x_i$ appears in both $x_{1,i}+x_1x_i=0$ and $x_{2,i}+x_2x_i=0$ then we define a new variable $x_i'$ and write
$$
\left\{
\begin{array}{lll}
x_{1,i}+x_1x_i&=&0\\
x_{2,i}+x_2x_i'&=&0\\
x_i+x_i'=&0&\;.
\end{array}
\right.
$$
By doing this for all shared variables, we add a set of $\nmq(\nmq-1)$ linear equations of the form $x_i'+x_i=0$, each one introducing a new variable.  
\\
With the first step we have produced a linear polynomial $f'$ with $\frac{\nmq(\nmq+1)}{2}$ variables. However, for the system to be in standard form, each linear equation has to involve at most three variables. As in Fact \ref{fact: linear standard}, this can be done in a third step by substituting $f'$ with $\frac{\nmq(\nmq+1)}{2}-2$ linear equations involving each one three variables, for a total of $2\left(\frac{\nmq(\nmq+1)}{2}\right)-3$ variables.
\\
We end up with a set of $\frac{\nmq(\nmq-1)}{2}$ quadratic equations and a set of $\frac{3\nmq^2-\nmq}{2}-2$ linear equations. The total number of variables is $\frac{5\nmq^2-\nmq}{2}-3$.

Now let $S$ contain $\mmq$ equations. We perform the same transformation as above for each of them, however the sets of quadratic equations of the polynomials might share variables. To solve this problem we rename the variables of each quadratic polynomial: if the $h$-th polynomial contains the variable $x_k$ then we substitute $x_k$ with a new variable $X_{h,k}$. We also need to track this substitution, therefore add a new set of linear equations $X_{h,k} + x_k = 0$ for $h=1,\ldots,\mmq$ and $k=1,\ldots,\nmq$.
In this way, the quadratic equations do not share variables, and we can substitute each one of them with a system in standard form, as we did in the first part of this proof. As a consequence, the total number of quadratic equations is bounded by $\mmq\left(\frac{\nmq(\nmq-1)}{2}\right)$ and the number of linear equations is at most $\mmq\left(\frac{3\nmq^2-\nmq}{2}-2\right)$. Similarly, the number of variables is bounded by $\mmq\left(\frac{5\nmq^2-\nmq}{2}-3\right)$. 

Putting everything together, we obtain $\mathcal{O}(\mmq\nmq^2)$ new variables and equations.	 
\end{proof}

\section{MLD to MQ reduction}\label{sec:MLD-MQ}
In this section we provide an explicit reduction $\alpha$, which maps an instance $I_{\mathrm{MLD}}$ of the MLD problem to an instance $I_{\mathrm{MQ}}$ of the MQ problem. 
More precisely, for any pair of complexity parameters $(n,m)$ we are going to define a reduction $\alpha_{n,m}$, which deals with binary codes of length $n$ and dimension at least $n-m$.
\\
An MLD instance $I_{\mathrm{MLD}} = (\bar{H},\bar{s},\bar{t})$ can be thought of as the union of two requirements:
\begin{enumerate}
    \item \textit{parity-check constraint}; the solution $\bar{v}$ has to satisfy $\bar{H}\bar{v}^{\top}=\bar{s}^{\top}$;
    \item \textit{weight constraint}; the solution $\bar{v}$ has to satisfy $\mathrm{w}(\bar{v})\leq \bar{t}$. To obtain our reduction we split this constraint in two parts: $\mathrm{w}(\bar{v})=w$ and $w\leq \bar{t}$. 
\end{enumerate}

We propose three encodings, each one parametrised by the complexity parameters $m$ and $n$ of MLD, which together correspond to a reduction from MLD to MQ.
Here the term \textit{encoding} has nothing to do with the mapping of a message to a codeword, instead it is the the rewriting of a constraint in terms of quadratic equations.
The set of quadratic Boolean polynomials $\pcce_{n,m}$ is the encoding of the parity-check constraint $\bar{H}v^\top=\bar{s}^\top$. A complete description is provided in Section \ref{subsec:paritycheck}. The polynomial system $\hwce_{n,m}$, detailed in Section \ref{subsec:wcomp}, corresponds to the Hamming weight computation of $\bar{v}$. The third encoding $\wce_{n,m}$, in Section \ref{subsec:wcheck}, is a polynomial system corresponding to the weight constraint $\mathrm{w}(v)\leq \bar{t}$.

We define the map $\alpha_{n,m}=\pcce_{n,m} \cup \hwce_{n,m}\cup\wce_{n,m}$, where we mean that, given a specific instance $I_{\mathrm{MLD}}= (\bar{H},\bar{s},\bar{t})$, the actual reduction is given by $I_{\mathrm{MQ}}=\alpha_{n,m}(I_{\mathrm{MLD}})=\pcce_{n,m}(\bar{H},\bar{s},\bar{t}) \cup \hwce_{n,m}(\bar{H},\bar{s},\bar{t}) \cup \wce_{n,m}(\bar{H},\bar{s},\bar{t})$.
\\
Observe that $\alpha_{n,m}$ is a system of polynomial equations depending only on $n$ and $m$, and whose evaluation on a specific MLD instance gives us an MQ instance.

\subsection{Parity-check constraint encoding}\label{subsec:paritycheck}
We claim that the parity-check matrix constraint $Hv^{\top}=s$ is equivalent to a set of $m$ linear equations corresponding to polynomials $$f_i \in \mathbb{F}[h_{i,1},\ldots,h_{i,n},v_1,\ldots,v_n, s_i]$$ of the form $f_i=\sum_{j=1}^n h_{i,j}v_j + s_i$. 
Indeed, $\bar{H}\bar{v}^{\top}=\bar{s}$ if and only if $f_i(\bar{H},\bar{v},\bar{s})=0$ for every $1\leq i\leq m$.
Thus, we can define $\pcce_{n,m}$ as
\begin{equation}\nonumber
\lbrace f_i=0 \rbrace_{i=1,\ldots, m}\;,
\end{equation}
and so
$$
\pcce_{n,m}(\bar{H},\bar{s},\bar{t})=\{f_i(\bar{H},v,\bar{s})=0\}_{i=1,\ldots,m}\;.
$$
Observe that $f_i(\bar{H},v,\bar{s})$ belongs to $\mathbb{F}[v_1,\ldots,v_n]$.

We explicitly state the following trivial result for completeness.
\begin{Lemma}\label{lem: complexity pcce}
Let $I_{\mathrm{MLD}}$ be an instance with complexity parameters $n$ and $m$.
Then, $\pcce_{n,m}(I_{\mathrm{MLD}})$ contains $m$ linear equations in $n$ variables.
\end{Lemma}

\subsection{Weight-computation encoding}\label{subsec:wcomp}
Let $\bar{v} \in  \mathbb{F}^n$ and $ \ell = \lfloor \mathrm{log}_2(n) \rfloor + 1$, so that the weight of a length-$n$ vector can be written as a length-$\ell$ vector. 
For $i=0,\ldots,n$ and $j=1,\ldots,\ell$ we want to define some functions $a^{(i)}:\mathbb{F}^n\to\mathbb{F}^{\ell}$ and their component functions $a^{(i)}(\bar{v}) = (a^{(i)}_1(\bar{v}),\ldots,a^{(i)}_{\ell}(\bar{v})) \in \mathbb{F}^{\ell}$. 
We set $a^{(0)}(\bar{v}) = (0,\ldots,0)$ for any $\bar{v}$ (for convenience), and we define $a^{(i)}$ recursively for $i=1,\ldots,n$ by computing its coefficients $a^{(i)}_j$ as the polynomials in $\mathbb{F}[v_1,\ldots,v_n]=\mathbb{F}[v]$
\begin{equation}\nonumber
a^{(i)}_j(v) = a^{(i-1)}_j(v) + \left(\prod_{h=1}^{j-1}a^{(i-1)}_h(v) \right)v_i\;.
\end{equation}
The following lemma and theorem prove that $a^{(i)}(\bar{v})$ is the implementation of a counter that stores the Hamming weight of $\bar{v}$ until the \textit{i-th} coordinate. 

\begin{Lemma}\label{lem:counter} Let $0\leq k\leq n-1$.
If $\bar{v}_{k+1} =1$ then $a^{(k+1)}(\bar{v}) = \increase(a^{(k)}(\bar{v}))$. If $\bar{v}_{k+1} =0$ then $a^{(k+1)}(\bar{v})=a^{(k)}(\bar{v})$.
\end{Lemma}
\begin{proof}
Let $a^{(k)}_z(\bar{v})$ be the leftmost 0 in $a^{(k)}(\bar{v})$, $1\le z\le \ell$, and let $\Lambda \in \mathbb{N}$. Observe that
\begin{equation}\nonumber
\prod_{h=1}^{\Lambda}a^{(k)}_h =  \begin{cases}
    1 & \text{if } \Lambda<z\\
    0 & \text{ otherwise}.
  \end{cases}
\end{equation}

We start with the case $\bar{v}_{k+1}=1$ and we compute $a^{(k+1)}_j(\bar{v})$ according to the value of $j$. 
\begin{itemize}
    \item if $1\leq j\leq z-1$ then $a^{(k+1)}_j(\bar{v}) =a^{(k)}_j(\bar{v}) + \left(\prod_{h=1}^{j-1}a^{(k)}_h(\bar{v}) \right)\bar{v}_{k+1} = 1 + (1) \cdot 1=0$;
    \item if $j=z$ we have $a^{(k+1)}_z(\bar{v}) =a^{(k)}_z(\bar{v}) + \left(\prod_{h=1}^{z-1}a^{(k)}_h(\bar{v}) \right)\cdot 1 = 0+(1)\cdot1=1$;
    \item if $j\ge z+1$ then $a^{(k+1)}_j(\bar{v}) =a^{(k)}_j(\bar{v}) + \left(\prod_{h=1}^{j-1}a^{(k)}_h(\bar{v}) \right)\cdot 1 = a^{(k)}_j(\bar{v}) + (0) \cdot 1=a^{(k)}_j(\bar{v})$.
\end{itemize}
The procedure we described flips every 1-bit (until the $(z-1)$-th bit) and the first 0-bit, while leaving all the other bits unchanged, which is the behaviour of the function $\increase()$.

The second possible case is $\bar{v}_{k+1}=0$. Observe that, regardless of $j$, the value $a_j^{(k+1)}(\bar{v})$ is given by
$$a^{(k+1)}_j(\bar{v}) =a^{(k)}_j(\bar{v}) + \left(\prod_{h=1}^{j-1}a^{(k)}_h(\bar{v}) \right)\bar{v}_{k+1}\;,$$
which, since $\bar{v}_{k+1}=0$, is simply
$$a^{(k+1)}_j(\bar{v}) =a^{(k)}_j(\bar{v}) + \left(\prod_{h=1}^{j-1}a^{(k)}_h(\bar{v}) \right)\cdot0=a^{(k)}_j(\bar{v})\;.$$
\end{proof}

\begin{theorem}\label{thm:w=int(a)}
Let $\bar{v} \in \mathbb{F}^n$ and let $1\leq i\leq n$. Then
\begin{equation}\nonumber
\mathrm{w}(\pi_i(\bar{v})) = \bintoint\left(a^{(i)}_1(\bar{v}), \ldots, a^{(i)}_{\ell}(\bar{v})\right)\;.
\end{equation}
\end{theorem}
\begin{proof}
We proceed by induction on $i$. Let $i=1$, then $a^{(1)}_1(\bar{v}) = a^{(0)}_1(\bar{v}) + \bar{v}_1 = \bar{v}_1$, while $a^{(1)}_j(\bar{v}) = 0+0\cdot0=0$ for every $j=2,\ldots,\ell$. Therefore $\bintoint(a^{(1)}(\bar{v}))=\bintoint(\bar{v}_1,0,\ldots,0) = \mathrm{w}(\pi_1(\bar{v}))$.
\\
Assume that $\bintoint(a^{(i)}(\bar{v})) = \mathrm{w}(\pi_i(\bar{v}))$ holds for a certain value of $1\leq i\leq n-1$, then we are going to prove $\bintoint(a^{(i+1)}(\bar{v})) = \mathrm{w}(\pi_{i+1}(\bar{v}))$. We have two cases: either $\bar{v}_{i+1} = 0$ or $\bar{v}_{i+1} = 1$.
\begin{enumerate}
\item If $\bar{v}_{i+1} = 0$ then by Lemma \ref{lem:counter} $a^{(i+1)}(\bar{v}) = a^{(i)}(\bar{v})$, in accordance with $\pi_{i+1}(\bar{v})= \pi_i(\bar{v})$.

\item If $\bar{v}_{i+1} = 1$ then, by Lemma \ref{lem:counter}, we have  $a^{(i+1)}(\bar{v}) = \increase(a^{(i)}(\bar{v}))$ and therefore $\bintoint(a^{(i+1)}(\bar{v})) = \bintoint(\increase(a^{(i)}(\bar{v}))) = \bintoint(a^{(i)}(\bar{v})) + 1 = \mathrm{w}(\pi_{i}(\bar{v})) +1 = \mathrm{w}(\pi_{i+1}(\bar{v}))$, where the last equality comes from obvious properties of weights.
\end{enumerate}
\end{proof}

The following corollary states that the Hamming weight of a vector $\bar{v} \in \mathbb{F}^n$ is exactly the integer represented by the vector $\left(a^{(n)}_1(\bar{v}), \ldots, a^{(n)}_{\ell}(\bar{v})\right)$.

\begin{Cor}\label{cor: weight bar a}
Let $\bar{v} \in \mathbb{F}^n$ then
\begin{equation}\label{cor:weightmap}
\mathrm{w}(\bar{v}) = \bintoint\left(a^{(n)}_1(\bar{v}), \ldots, a^{(n)}_{\ell}(\bar{v})\right)
\end{equation}
\end{Cor}
\begin{proof}
It follows from Theorem \ref{thm:w=int(a)} by noticing that $\pi_n(\bar{v}) = (\bar{v}_1,\ldots,\bar{v}_n) = \bar{v}$.
\end{proof}
Observe that, for $1\leq j\leq \ell$ and for any $1\leq i\leq n$, Formula \eqref{cor:weightmap} is an equation of degree $j$ of the form $x+y+M=0$ where $M$ is a degree-$j$ monomial. By using the procedure described in Remark \ref{rem:quadratize} each equation can be reduced to a system of less than $\ell$ quadratic equations by adding less than $\ell$ new variables. Indeed, $\quadr(x+y+M=0)$ is a quadratic system with $j-1<\ell$ equations in $j-2< \ell$ variables.
\\
By performing this procedure to each of the $n\cdot \ell$ equations we obtain a system of quadratic equations. We remark that many variables were shared between equations, so there are several possible optimizations to be applied instead of applying $\quadr()$ to each Equation \eqref{cor:weightmap}. However, since the degree of each equation \eqref{cor:weightmap} is bounded by $\ell$, even without optimizing the procedure, we end up with a system with a manageable number of quadratic equations, as stated by the following lemma.
\begin{Lemma}\label{lem: complexity hwce}
The number of quadratic equations in
$$
\quadr\left(
\left\lbrace a^{(i)}_j = a^{(i-1)}_j + \left(\prod_{h=1}^{j-1}a^{(i-1)}_h \right)v_i \right\rbrace_{
i = 1,\ldots,n \textit{, }
j=1,\ldots,\ell}
\right)
$$
is $\mathcal{O}(n\ell^2$).
The total number of variables is $\mathcal{O}(n\ell^2)$.
\end{Lemma}
\begin{proof}
Each of the $n\ell$ equations is transformed into a set of less than $\ell$ quadratic equations by adding less than $\ell$ variables. 
\end{proof}

We are finally ready to construct the set of equations corresponding to the (Hamming) weight-computation encoding, which contains the polynomials corresponding to the weight-computation procedure depicted in this section.
We define
\begin{equation}\nonumber
\hwce_{n,m} = \quadr\left(
\left\lbrace a^{(i)}_j = a^{(i-1)}_j + \left(\prod_{h=1}^{j-1}a^{(i-1)}_h \right)v_i \right\rbrace_{
i = 1,\ldots,n \textit{, }
j=1,\ldots,\ell}
\right)\;,
\end{equation}
where the variables $a_j^{(i)}$ obviously play the role of the previously defined functions.
From Theorem \ref{thm:w=int(a)} and Corollary \ref{cor: weight bar a}, it follows that by evaluating $\hwce_{n,m}$ at $\bar{v}$ we obtain the vector $(a_1^{(n)}(\bar{v}),\ldots,a_{\ell}^{(n)}(\bar{v}))$ containing the binary expansion of $\mathrm{w}(\bar{v})$.

\subsection{Weight constraint encoding}\label{subsec:wcheck}
The idea for the construction of this encoding is the definition of a quadratic Boolean polynomial capable of comparing two Boolean vectors according to the values of these two vectors when seen as integers. More precisely, in this section we construct a polynomial whose evaluation at a pair of Boolean vectors $\bar{u}$ and $\bar{v}$ is $0$ if and only if $\bintoint(\bar{u})\leq \bintoint(\bar{v})$. For the sake of completeness we define the claimed polynomial for a more general situation, and in the end we specify the parameters useful in the context of our reduction from MLD to MQ.

Consider two binary vectors $\bar{u}, \bar{v} \in \mathbb{F}^l$, where the most significant bits are $\bar{u}_l$ and $\bar{v}_l$, respectively. To compare the integers associated to $\bar{u}$ and $\bar{v}$ we can follow the following procedure, which outputs $0$ when $\bintoint(\bar{u}) \leq \bintoint(\bar{v})$ and $1$ otherwise, starting from $j=l$:
\begin{itemize}
    \item if $\bar{u}_j=\bar{v}_j$, then we move to the next bits $\bar{u}_{j-1}$ and $\bar{v}_{j-1}$;
    \item if $\bar{u}_j\neq \bar{v}_j$, we output $\bar{u}_j$, since $\bintoint(\bar{u}) \leq \bintoint(\bar{v})$ if $\bar{u}_j=0$;
    \item if we reach $j=1$ and $\bar{u}_1=\bar{v}_1$, then we output $0$.
\end{itemize}
We make use of this procedure to define a polynomial $F \in \mathbb{F}[u_1,\ldots,u_l,v_1,\ldots,v_l]$ such that
\begin{equation}\label{eq: function F}
F(\bar{u},\bar{v}) = \begin{cases}
0 	\quad\mathrm{ if }\; \bintoint(\bar{u}) \leq \bintoint(\bar{v})\\
1 \quad\mathrm{ if }\; \bintoint(\bar{u}) > \bintoint(\bar{v})\;.
\end{cases}
\end{equation}

Define $g_h(u,v) = (u_h + v_h) \in \mathbb{F}[u_h,v_h]$ for every $h=1,\ldots,l$ and notice that $g_h(\bar{u},\bar{v})=0$ if and only if $\bar{u}_h=\bar{v}_h$. Moreover, for $j=1,\ldots,l$, define the polynomials
\begin{equation}\label{eq:locators}
f_j = g_j\prod_{h=j+1}^l(g_h + 1)\;,
\end{equation}
where $f_j \in  \mathbb{F}[u_1,\ldots,u_l,v_1,\ldots,v_l]$. Clearly in the special case $j=l$ we have $f_l=g_l$. Observe that the degree of $f_j$ is $l-j+1$.

Given two vectors $\bar{u}, \bar{v} \in \mathbb{F}^l$, the purpose of the set of polynomials $\lbrace f_j \rbrace_{j=1}^l$ in \eqref{eq:locators} is to locate the most significant bit in which $\bar{u}$ and $\bar{v}$ differ. We prove this in the following lemmas.

\begin{Lemma}\label{lem:atmostlocation}
Let $\bar{u}, \bar{v} \in \mathbb{F}^l$, then $f_j(\bar{u}, \bar{v}) = 1$ for at most one value of $j$.
\end{Lemma}
\begin{proof}
Assume by contradiction that $f_j(\bar{u}, \bar{v}) = f_i(\bar{u}, \bar{v}) = 1$ for $i \ne j$. We can assume, without loss of generality, that $i < j$. By construction of $f_j$ we must have that $g_j(\bar{u}, \bar{v}) = 1$. But $(g_j + 1) \mid f_i$ since $i < j$, $(g_j+1)(\bar{u},\bar{v})=0$, and therefore $f_i(\bar{u},\bar{v})=0$, which is a contradiction.

\end{proof}

\begin{Lemma}\label{lem:equality}
Let $\bar{u}, \bar{v} \in \mathbb{F}^l$, then $f_j(\bar{u}, \bar{v}) = 0 \textit{ } \forall \textit{ } j=1,\ldots,l$ if and only if  $\bar{u} = \bar{v}$.
\end{Lemma}
\begin{proof} 
We show the first implication.
Since $f_l(\bar{u},\bar{v}) = g_l(\bar{u},\bar{v})=0$, then $\bar{u}_l=\bar{v}_l$. 
\\
We claim that, for any $1\leq k\leq l-1$, $\bar{u}_{k}$ is equal to $\bar{v}_{k}$, provided that $\bar{u}_h = \bar{v}_h$ for any $k+1 \leq h \leq l$.
Notice that $(g_{h}+1) \mid f_{k}$ for every value $h = k+1,\ldots,l$. Since $\bar{u}_{h} = \bar{v}_{h}$ then $(g_{h} + 1)(\bar{u}, \bar{v}) = 1$, for $h = k+1,\ldots,l$. This implies $0 = f_{k}(\bar{u}, \bar{v}) =  g_{k}(\bar{u}, \bar{v})\cdot1=\bar{u}_k+\bar{v}_k$ and therefore $\bar{u}_{k} = \bar{v}_{k}$. 
\\
The first implication follows by an iterated application of our claim from $k=l-1$ until $k=1$.

As regards the second implication, $\bar{u} = \bar{v}$ implies $\bar{u}_j = \bar{v}_j$ for every $j = 1,\ldots,l$, and so $g_j(\bar{u}, \bar{v}) = 0$. Observe also that, by construction in (\ref{eq:locators}), $g_j \mid f_j$ for every $j = 1,\ldots,l$, which forces $f_j(\bar{u}, \bar{v})=0$.
\end{proof}

\begin{Lemma}\label{lem:locationunique}
Let $\bar{u}, \bar{v} \in \mathbb{F}^l$ with $\bar{u} \neq \bar{v}$, then there exists a unique $j$ such that $f_j(\bar{u},\bar{v}) = 1$. Moreover $\bar{u}_j\neq\bar{v}_j$ while $\bar{u}_k = \bar{v}_k$ for every $k = j+1, \ldots,l$.
\end{Lemma}

\begin{proof}
If $\bar{u} \neq \bar{v}$ then there exists $j$ such that $\bar{u}_j \neq \bar{v}_j$. Let $j$ be such that $\bar{u}_j \neq \bar{v}_j$ and $\bar{u}_k = \bar{v}_k$ for every $k>j$. This implies $(g_k +1)(\bar{u} , \bar{v}) = 1$ for each $k>j$, as well as $g_j(\bar{u} , \bar{v})=1$, thus $f_j(\bar{u} , \bar{v}) = 1$. The uniqueness of $j$ follows from Lemma \ref{lem:atmostlocation}.

For the second part of the statement, if there exists $k > j$ such that $\bar{u}_k \neq \bar{v}_k$, then $g_k(\bar{u}_k, \bar{v}_k) = 1$. However, $(g_k +1)\mid f_j$ and $(g_k + 1)(\bar{u}, \bar{v}) = 0$ imply  $f_j(\bar{u}, \bar{v}) = 0$, which is a contradiction.
\end{proof}

We are ready to define our function $F$ as in Equation \eqref{eq: function F}.
\begin{Prop}\label{prop:weightconst}
Let $F = \sum_{j=1}^l f_j\cdot (v_j+1) \in \mathbb{F}[u_1,\ldots,u_l,v_1,\ldots,v_l]$ and let $\bar{u}, \bar{v} \in \mathbb{F}^l$. Then $F(\bar{u},\bar{v}) = 0$ if and only if $\bintoint(\bar{u}) \le \bintoint(\bar{v})$.
\end{Prop}

\begin{proof}
We have two cases: either $f_j(\bar{u}, \bar{v}) = 0$ for every value $j = 1,\ldots,l$,
or there exists a (by Lemma \ref{lem:locationunique}) \textit{unique} $k$ such that $f_k(\bar{u},\bar{v}) = 1$.
\\
The first case, due to Lemma \ref{lem:equality}, is equivalent to $\bar{u}= \bar{v}$ and thus it corresponds to the equality $\bintoint(\bar{u})=\bintoint(\bar{v})$. By definition of $F$, the first case also implies $F(\bar{u},\bar{v})=0$.\\ We have proved that $\bintoint(\bar{u})=\bintoint(\bar{v})$ implies $F(\bar{u},\bar{v})=0$, while $F(\bar{u},\bar{v})=0$ implies either that $\bintoint(\bar{u})=\bintoint(\bar{v})$ or that we are not in the first case.

In the second case, let $f_k(\bar{u}, \bar{v}) = 1$ for a certain value $k$. Lemma \ref{lem:equality} implies that $\bintoint(\bar{u})\neq\bintoint(\bar{v})$, and by Lemma \ref{lem:locationunique} we have $\bar{u}_j=\bar{v}_j$ for $k+1\leq j\leq l$ and $\bar{u}_k\neq \bar{v}_k$. \\
If $\bar{u}_k=0$ and $\bar{v}_k=1$, then $\bintoint(\bar{u})<\bintoint(\bar{v})$ and
$$
F(\bar{u},\bar{v}) = \sum_{j\neq k}f_j(\bar{u},\bar{v})\cdot (\bar{v}_j +1)+f_k(\bar{u},\bar{v})\cdot(\bar{v}_k+1) =
\sum_{j\neq k}0\cdot (\bar{v}_j +1)+1\cdot(\bar{v}_k+1)=\bar{v}_k+1=0\;.
$$
If instead $\bar{u}_k=1$ and $\bar{v}_k=0$, then $\bintoint(\bar{u})>\bintoint(\bar{v})$ and
$$
F(\bar{u},\bar{v}) = \sum_{j\neq k}f_j(\bar{u},\bar{v})\cdot (\bar{v}_j +1)+f_k(\bar{u},\bar{v})\cdot(\bar{v}_k+1) =
\sum_{j\neq k}0\cdot (\bar{v}_j +1)+1\cdot(\bar{v}_k+1)=\bar{v}_k+1=1\;.
$$
Either way, if $\bintoint(\bar{u})\neq\bintoint(\bar{v})$ then $F(\bar{u},\bar{v})=0$ if and only if $\bintoint(\bar{u}) < \bintoint(\bar{v})$.

\end{proof}
Observe that the degree of $F$ is equal to $\deg(f_1)+1=l+1$. Observe also that the degree of $f$ evaluated at $\bar{v}_1,\ldots,\bar{v}_l$, which is a polynomial in $\mathbb{F}[u_1,\ldots,u_l]$, is at most $l$.
\\

Let $I_{\mathrm{MLD}}=(\bar{H},\bar{s},\bar{t})$, let $\bar{v}$ be a vector for which $\bar{H}\bar{v}^{\top}=\bar{s}^{\top}$ and let $\bar{t}=(\bar{t}_1,\ldots,\bar{t}_{\ell})$. 
We recall that in Section \ref{subsec:wcomp} we defined a set of $\ell$ Boolean functions $a_1^{(n)},\ldots,a_{\ell}^{(n)}$ representing the weight of a length-$n$ vector $v$, i.e. by Corollary \ref{cor: weight bar a} $\bintoint(a_1^{(n)}(\bar{v}),\ldots,a_{\ell}^{(n)}(\bar{v}))=\mathrm{w}(\bar{v})$.

Then, by Corollary \ref{cor: weight bar a} and Proposition \ref{prop:weightconst},
$$
F(a_1^{(n)}(\bar{v}),\ldots,a_{\ell}^{(n)}(\bar{v}),\bar{t}_1,\ldots,\bar{t}_{\ell})=0
$$
if and only if $$
\mathrm{w}(\bar{v})\leq \bintoint(\bar{t})\;.
$$
Let us consider the following polynomial
$$
F_{\ell}\in\mathbb{F}[a_1^{(n)},\ldots,a_{\ell}^{(n)},t_1,\ldots,t_{\ell}]\;,
$$
which is obtained by rewriting $F$ in the variables $a_1^{(n)},\ldots,a_{\ell}^{(n)}$ and in (the new variables of) $t_1,\ldots,t_{\ell}$.
Clearly, $F_{\ell}(a_1^{(n)},\ldots,a_{\ell}^{(n)},t_1,\ldots,t_{\ell})=0$ encodes the constraint  
$\mathrm{w}(\bar{v})\leq \bintoint(\bar{t})$.
Obviously, the degree of $F_{\ell}(a_1^{(n)},\ldots,a_{\ell}^{(n)},\bar{t}_1,\ldots,\bar{t}_{\ell})\in\mathbb{F}[a_1^{(n)},\ldots,a_{\ell}^{(n)}]$ is at most $\ell$.
By applying the map $\quadr()$ to $F_{\ell}$ (see Remark \ref{rem:quadratize}) we obtain the system of quadratic equations 
$$
\wce_{n,m} = \quadr\left(\lbrace F_{\ell} \rbrace\right)\;.
$$
\begin{Lemma}\label{lem: complexity wce}
$\wce_{n,m}(I_{\mathrm{MLD}})$ is a system of $\mathcal{O}(n\ell)$ quadratic equations in $\mathcal{O}(n\ell)$ variables.
\end{Lemma}
\begin{proof}
Since the degree of $F_{\ell}$ is at most $\ell$ and the size of its support, i.e. the number of monomials, is at most $2^{\ell}$, by applying $\quadr()$ we obtain a system with at most $2^{\ell}\ell$ equations in $2^{\ell}\ell$ variables.
However, $2^{\ell}\leq2n$, and so the system contains $\mathcal{O}(n\ell)$ equations in $\mathcal{O}(n\ell)$ variables.
\end{proof}

\subsection{MLD to MQ}
By combining the results of Sections \ref{subsec:paritycheck}, \ref{subsec:wcomp} and \ref{subsec:wcheck}, we construct the reduction $\alpha_{n,m}$, a function mapping MLD instances to MQ instances.
\begin{theorem}
Let $I_{\mathrm{MLD}} = (\bar{H},\bar{s},\bar{t})$ and let $\alpha_{n,m}:\mathrm{MLD} \rightarrow \mathrm{MQ}$ be defined by
$$
 \alpha_{n,m}(\bar{H},\bar{s},\bar{t}) = \pcce_{n,m}(\bar{H},\bar{s},\bar{t}) \cup \hwce_{n,m}(\bar{H},\bar{s},\bar{t}) \cup \wce_{n,m}(\bar{H},\bar{s},\bar{t})\, ;
 $$
  If $\bar{u}$ is a witness of $\alpha_{n,m}(I_{\mathrm{MLD}})$ then it is also a witness of $I_{\mathrm{MLD}}$.
\end{theorem}
\begin{proof}
By ordering the variables according to their first appearance in this paper, the first $n$ variables in $\alpha_{n,m}(I_{\mathrm{MLD}})$ are $v_1,\ldots,v_n$ and the first $n$ bits of a witness $u$ for $\alpha_{n,m}(I_{\mathrm{MLD}})$ are the values $\bar{v}_1,\ldots,\bar{v}_n$. Let us call $\bar{v}=(\bar{v}_1,\ldots,\bar{v}_n)$.\\
The set $\pcce_{n,m}(\bar{H},\bar{s},\bar{t})$ contains only equations in the variables $v_1,\ldots,v_n$.
Therefore, since $\bar{u}$ is a witness for $\alpha_{n,m}(I_{\mathrm{MLD}})$, we have $f(\bar{u})=0$ for $f \in \pcce_{n,m}(\bar{H},\bar{s},\bar{t})$, meaning that $\bar{H}\bar{v}^{\top}= \bar{s}$.
\\
From Corollary \ref{cor: weight bar a}, we can use the polynomials in $\hwce_{n,m}(\bar{H},\bar{s},\bar{t})$ to compute a binary expansion of the weight of $\mathrm{w}(\bar{v})$. More precisely, we have that $$\bintoint\left(a^{(n)}_1(\bar{v}), \ldots, a^{(n)}_l(\bar{v})\right) = \mathrm{w}(\bar{v})\;.$$
\\
Finally,  since $F_{\ell}\left(a^{(n)}_1(\bar{v}), \ldots, a^{(n)}_l(\bar{v}), \bar{t}\right) = 0$, by Proposition \ref{prop:weightconst} we have  $\mathrm{w}(\bar{v}) \le \bintoint(\bar{t})$.
\end{proof}

\begin{theorem}
Let $(n,m)$ be the complexity parameters of $I_{\mathrm{MLD}}$ and let $I_{\mathrm{MQ}}=\alpha_{n,m}(I_{\mathrm{MLD}})$. Then, the complexity parameters $(\nmq,\mmq)$ of $I_{\mathrm{MQ}}$ are in $\mathcal{O}(n\log_2^2n)$.
\end{theorem}
\begin{proof}
We analyse the 3 sub-problems separately. 
\begin{itemize}
    \item According to Lemma \ref{lem: complexity pcce}, $\pcce_{n,m}(I_{\mathrm{MLD}})$ is a linear system with $m$ equations and $n$ variables. 
    \item As described by Lemma \ref{lem: complexity hwce}, $\hwce_{n,m}(I_{\mathrm{MLD}})$ is a quadratic system of $\mathcal{O}(n\ell^2)$ equations in $\mathcal{O}(n\ell^2)$ variables.
    \item In the weight-constraint step described in Subsection \ref{subsec:wcheck}, we introduce a quadratic system $\wce_{n,m}(I_{\mathrm{MLD}})$ containing $\mathcal{O}(n\ell)$ equations and variables, as stated in Lemma \ref{lem: complexity wce}.
\end{itemize}
By putting everything together we can compute the complexity parameters $(\nmq,\mmq)$ of $\alpha_{n,m}(I_{\mathrm{MLD}})$.

\end{proof}

\section{MQ to MLD reduction}\label{sec:MQ-MLD}
In this section we construct a reduction $\beta:\mathrm{MQ}\to\mathrm{MLD}$. We consider a general system of equations and we take it to standard form $S$, as defined in Definition \ref{def:standardform}, by using the procedure hinted in the proof of Lemma \ref{lem:systostandard}. The result of the process will be an instance $\beta(I_{\mathrm{MQ}}) = (H,s,t)\in\mathrm{MLD}$, with $H \in \mathbb{F}^{m\times n}$ for some $m,n,t \in \mathbb{Z}^+$, $s \in \mathbb{F}^m$. 
The key idea is to regard the two different types of equations in $S$ separately. First we create one part of the MLD instance according to the quadratic equations in $S$, and then we complete the work by integrating the linear ones. 
We also build a transformation that takes as input a solution of $\beta(I_{\mathrm{MQ}})$ and outputs a solution of $I_{\mathrm{MQ}}$, proving that $\beta$ is actually a reduction between the two problems.

\subsection{Quadratic equations}\label{sec: quadratic}
Consider a system $S$ containing solely the equation $xy+z =0$ and let $I= \langle xy+z\rangle \subset \mathbb{F}[x,y,z]$ be the principal ideal generated by $S$. The associated variety $\mathcal{V}_{\mathbb{F}}(I) \subset \mathbb{F}^3$ is 
\begin{equation}\label{eq:variety}
\mathcal{V}_{\mathbb{F}}(I) = \lbrace (0,0,0), (1,0,0),(0,1,0),(1,1,1)\rbrace.
\end{equation}
\begin{Lemma}\label{lem:sighat}
Let $\widehat{C}$ be the linear code generated by the generator matrix 
\begin{equation}\label{eq: def G}
\widehat{G} = \begin{bmatrix}
1&0&0&1&1&0&0&1&1&1\\
0&1&0&0&0&1&1&1&1&1\\
0&0&1&1&1&1&1&1&1&1
\end{bmatrix}\;,
\end{equation}
let $$
\widehat{\epsilon} = (0,0,0,0,0,0,0,1,1,1)\;,
$$
let $\widehat{\Sigma}$ be the coset $\widehat{\epsilon}+\widehat{C}\subseteq \mathbb{F}^{10}$ and let $v\in\widehat{\Sigma}$.
\\
Then, $v$ has weight at most $3$ if and only if $\tau_3(v)\in\mathcal{V}_{\mathbb{F}}(I)$ as in \eqref{eq:variety}.
\end{Lemma}
\begin{proof}
It follows by direct inspection of the $8$ vectors in $\widehat{\Sigma}$.
\end{proof}
The truncation map present in the statement of the previous Lemma, $\tau_3 : \mathbb{F}^{10} \rightarrow \mathbb{F}^{3}$ defined as $\tau_3(\bar{v}_1,\bar{v}_2,\bar{v}_3,\bar{v}_4,\bar{v}_5,\bar{v}_6,\bar{v}_7,\bar{v}_8,\bar{v}_9,\bar{v}_{10}) = (\bar{v}_1, \bar{v}_2, \bar{v}_3)$, 
	can be represented in matrix form as
	\begin{equation}\nonumber
		M_{\tau_3} = 
		\left[\begin{array}{c} 
	\mathbb{I}_3\\ 
	\hline 
	\mathbf{0}
\end{array}\right]
	\;,
	\end{equation}
	where $\mathbb{I}_3$ is the identity matrix of dimension $3$ and $\mathbf{0}$ is the $3 \times 7$ zero matrix, namely
	$$
	\tau_3(\bar{v})=\bar{v}M_{\tau_3}\;.
	$$

\begin{Prop}\label{prop: singlequadeq}
Let $\widehat{H}$ be a parity-check matrix for the code $\widehat{C}$ generated by $\widehat{G}$ as defined in Lemma \ref{lem:sighat}, let $\widehat{s}=\widehat{H}\cdot \widehat{\epsilon}^{\top}$ and let $\widehat{t}=3$. Let $\widehat{W}\subset \widehat{\Sigma}$ be the set of witnesses of the MLD instance $I_{\mathrm{MLD}}=(\,\widehat{H},\widehat{s},\widehat{t}\,)$ and let $\mathcal{V}_{\mathbb{F}}(I)$ be as in \eqref{eq:variety}. Then $\mathcal{V}_{\mathbb{F}}(I)=\tau_3(\widehat{W})$.
\end{Prop}

\begin{proof}
It follows by Lemma \ref{lem:sighat} and the well-known bijection between cosets and syndromes (once the parity-check matrix has been chosen).
\end{proof}

\begin{Rem}\label{rem:weightgap}
The witnesses of the MLD instance $(\,\widehat{H},\widehat{s},\widehat{t}\,)$ we constructed, i.e. solutions of $\widehat{H}v^{\top}=\widehat{s}^{\top}$ with weight at most $\widehat{t}=3$, have Hamming weight exactly $3$, which is the weight of any coset leader (e.g. $\widehat{\epsilon}$). Notice that the remaining solutions of $\widehat{H}v^{\top}=\widehat{s}^{\top}$ have weight at least $5$. This gap in the weight is crucial for the generalization we are going to give next. 
\end{Rem}

We now extend the construction in Proposition \ref{prop: singlequadeq} to a standard-form system $S$ that contains more than one quadratic equation. Assume that $S$ contains $q$ quadratic equations $f_i$'s of the form $x_iy_i+z_i=0$,  for $i = 1,\ldots,q$. Recall that by Definition \ref{def:standardform}, such equations do not share any variable with each other. Consider the ideal $J = \langle f_1, \ldots f_q \rangle \subset \mathbb{F}[\{x_i,y_i,z_i\}_{i=1,\ldots,q}]$. The variety $\mathcal{V}_{\mathbb{F}}(J)$ can be seen as
\begin{equation}\label{eq:varietyext}
\mathcal{V}_{\mathbb{F}}(J) = \underbrace{\mathcal{V}_{\mathbb{F}}(I) \times \mathcal{V}_{\mathbb{F}}(I) \times \cdots \times \mathcal{V}_{\mathbb{F}}(I)  }_\text{q} \subset \mathbb{F}^{3q}\;,
\end{equation}
where $\mathcal{V}_{\mathbb{F}}(I)$ is as in \eqref{eq:variety}.
To address the case of standard-form systems consisting of only quadratic equations, we construct a new parity-check matrix as the diagonal block matrix $\widetilde{H}$ of size $7q\times 10q$
\begin{equation}\nonumber
\widetilde{H} = \begin{bmatrix}
\widehat{H} & \cdots & 0\\
\vdots & \ddots & \vdots\\
0 & \cdots & \widehat{H}
\end{bmatrix}\;,
\end{equation}
where $\widehat{H}$ is a parity-check matrix for the code generated by $\widehat{G}$ in Equation \eqref{eq: def G}.
Obviously, the null space of $\widetilde{H}$ is the direct product of $q$ copies of $\widehat{C}$ i.e. 
\begin{equation}\nonumber
\underbrace{\widehat{C} \times \widehat{C} \times \cdots \times \widehat{C} }_\text{q} \subset \mathbb{F}^{10q}\;.
\end{equation}
Define $\widetilde{\Sigma} = \overbrace{\widehat{\Sigma} \times \widehat{\Sigma} \times \cdots \times \widehat{\Sigma} }^\text{q}$, $\widetilde{t} = 3q$ and $\widetilde{\epsilon} = (\,\overbrace{\widehat{\epsilon} \| \widehat{\epsilon} \| \cdots \| \widehat{\epsilon}}^\text{q}\,)$ where $\|$ denotes vector concatenation. With this setting, for any $\widetilde{v}\in\widetilde{\Sigma}$, we can write $\widetilde{v}  = \left(\widehat{v}^{(1)} \| \widehat{v}^{(2)} \| \cdots \| \widehat{v}^{(q)}\right)$ with $\widehat{v}^{(i)} \in \widehat{\Sigma}$ for any $i$.
In the following lemma $\widehat{W}$ is the set of vectors in $\widehat{\Sigma}$ with weight at most $\widehat{t}$, as defined in Proposition \ref{prop: singlequadeq}.

\begin{Lemma}\label{lem:sigmacomp}
	Let $\widetilde{W} = \lbrace \widehat{v} \in \widetilde{\Sigma} \mid \mathrm{w}(\widehat{v}) \le \widetilde{t} \rbrace $. Then $\widetilde{W} = \underbrace{\widehat{W}\times \widehat{W}\times \cdots\times \widehat{W}}_\text{q}$.
\end{Lemma}
\begin{proof}
	Notice that, by Remark \ref{rem:weightgap}, $\mathrm{w}(\widehat{v})=\widehat{t}= 3$ for every $\widehat{v} \in \widehat{W}$. So $\mathrm{w}(\widetilde{v}) = \tilde{t} = 3q$ for every $\widetilde{v} \in \widehat{W}\times \widehat{W}\times \cdots\times \widehat{W}$. Therefore $\widetilde{W} \supseteq \widehat{W}\times \widehat{W}\times \cdots\times \widehat{W}$.
\\	
	To prove the other inclusion consider $\widetilde{v} 
	\in \widetilde{W} $.
	Then $\widetilde{v}$ can be written as a concatenation of $q$ vectors in $\widehat{\Sigma}$ and, by Remark \ref{rem:weightgap}, each of such vectors has Hamming weight at least $3$, and $\mathrm{w}(\widetilde{v}) =\sum_{i=1}^q\mathrm{w}(\widehat{v}^{(i)})$.
	If there is a $\widehat{v}^{(i)}$ with weight more than $3$, then the weight of $\widetilde{v}$ is strictly larger that $3q$. Therefore, all $\widehat{v}^{(i)}$ must have weight exactly $3$. This proves $\widetilde{W} \subseteq \widehat{W}\times \widehat{W}\times \cdots\times \widehat{W}$.
\end{proof}
Consider the truncation map $\tau: \mathbb{F}^{10q} \rightarrow \mathbb{F}^{3q}$ whose matrix representation is
\begin{equation}\label{eq def Mtau}
M_{\tau} = \begin{bmatrix}
	M_{\tau_3} & \cdots & 0\\
	\vdots & \ddots & \vdots \\
	0 & \cdots & M_{\tau_3}
\end{bmatrix}\;.
\end{equation}
The following lemma proves a useful property of $\tau$ that comes at hand to prove the subsequent theorem.
\begin{Lemma}\label{lem:projectionextended}
	Let $\widetilde{v} \in \widetilde{\Sigma}$ then $\tau(\widetilde{v}) = \left( \tau_3(\widehat{v}^{(1)}),\ldots, \tau_3(\widehat{v}^{(q)})\right)$ where $\widehat{v}^{(i)} \in \widehat{\Sigma}$ for every $i=1,\ldots,q$.
\end{Lemma}
\begin{proof}
	Considering the matrix representation of $\tau$, we obtain
	\begin{equation}\nonumber
	\begin{split}
		\tau(\widetilde{v}) &= \widetilde{v}M_{\tau} \\
		& =
\begin{pmatrix}
\widehat{v}^{(1)} \|
\cdots \|
\widehat{v}^{(q)}
\end{pmatrix}
		\begin{bmatrix}
	M_{\tau_3} & \cdots & 0\\
	\vdots & \ddots & \vdots \\
	0 & \cdots & M_{\tau_3}
\end{bmatrix}  \\ &= \left( \widehat{v}^{(1)}M_{\tau_3},\ldots,\widehat{v}^{(q)}M_{\tau_3}\right)\\
&= \left( \tau_3(\widehat{v}^{(1)}),\ldots, \tau_3(\widehat{v}^{(q)})\right)\;.
\end{split}
	\end{equation}
\end{proof}

 \begin{Prop}\label{prop: multiplequadeq}
 	Set $\widetilde{s} = \widetilde{H}\widetilde{\epsilon}^{\top}$, then $\widetilde{W}$ solves the MLD instance $\left( \widetilde{H},\widetilde{s},\widetilde{t} \right)$. Moreover $\mathcal{V}_{\mathbb{F}}(J) = \tau(\widetilde{W})$.
 \end{Prop}
 \begin{proof}
 We need to prove that given $\widetilde{v} \in \widetilde{W}$ it holds $\widetilde{H}\widetilde{v}^\top = \widetilde{s}$ and $\mathrm{w}(\widetilde{v}) \le \widetilde{t}$.
 The null space of $\widetilde{H}$ is 
 $\widetilde{H}^\bot = \widehat{C} \times \widehat{C} \times \cdots \times \widehat{C} $. 
 Observe that 
 \begin{equation}\nonumber
 \begin{split}
 \widetilde{\Sigma} &= \widehat{\Sigma} \times \widehat{\Sigma} \times \cdots \times \widehat{\Sigma} \\ &= \left( \widehat{C}+\widehat{\epsilon} \right) \times \left( \widehat{C}+\widehat{\epsilon} \right) \times \cdots \times \left( \widehat{C}+\widehat{\epsilon} \right) \\ 
 &= \left( \widehat{C} \times \widehat{C} \times \cdots \times \widehat{C}  \right) + \widetilde{\epsilon} \\
 &= \widetilde{H}^\bot + \widetilde{\epsilon} = \lbrace v + \widetilde{\epsilon} \mid v \in \widehat{C} \times \widehat{C} \times \cdots \times \widehat{C} \rbrace
 \end{split}
 \end{equation} 
 Therefore, for every $\widetilde{v} \in \widetilde{\Sigma}$, we have $\widetilde{H}^\bot\widetilde{v}^\top = \widetilde{H}^\bot \widetilde{v}^\top + \widetilde{H}^\bot \widetilde{\epsilon}^\top = \widetilde{s}$. Considering $\widetilde{v} \in \widetilde{W} \subset \widetilde{\Sigma}$ we also obtain $\mathrm{w}(\widetilde{v}) \le \widetilde{t}$.
 
 For the second claim we have, by lemmas \ref{lem:sigmacomp} and \ref{lem:projectionextended}, that
 \begin{equation}\nonumber
 	\begin{split}
 	\tau\left( \widetilde{W} \right) &= \tau \left( \widehat{W}\times \widehat{W}\times \cdots\times \widehat{W} \right)\\
 	&= \tau_3(\widehat{W})\times\ldots\times \tau_3(\widehat{W}) \\
 	&= \mathcal{V}_{\mathbb{F}}(I) \times \mathcal{V}_{\mathbb{F}}(I) \times \cdots \times \mathcal{V}_{\mathbb{F}}(I) \\
 	&= \mathcal{V}_{\mathbb{F}}(J)\;,
 	\end{split}
 \end{equation}
 where the last two equalities hold due to Proposition \ref{prop: singlequadeq} and Equation \eqref{eq:varietyext}.
 \end{proof}

\subsection{Linear equations}

	Let $S$ be a standard-form system containing $q$ quadratic equations. Due to definition \ref{def:standardform}, $S$ is a system in exactly $3q$ variables. We can thus write $S \subset \mathbb{F}[x_1, \ldots, x_{3q}]$.

\begin{Rem}\label{rem:invproj}
Consider a linear polynomial $f$ in $\mathbb{F}[x_1,\ldots,x_{3q}]$ for some value of $q \in \mathbb{Z}^+$. We can write $f = \sum_{i=1}^{3q}a_ix_i + \delta$ and define the vector of its coefficients $a_f = (a_1,\ldots,a_{3q}) \in \mathbb{F}^{3q}$. Notice that the vector $a_f$ contains only the coefficients of $x_1,\ldots,x_{3q}$ and not the term $\delta$. 
With this notation we observe that $\bar{w} \in \mathbb{F}^{3q}$ belongs to $\mathcal{V}_{\mathbb{F}}(\langle f \rangle)$ if and only if the product $\bar{w}\cdot a_f^\top = \delta$.
\end{Rem}
The reduction introduced in Section \ref{sec: quadratic} deals with standard-form systems that include only quadratic equations. This reduction is formalised in Proposition \ref{prop: multiplequadeq} as a map taking as input a system of $q$ equations in $3q$ variables and outputting an MLD instance corresponding to a $3q\times 10q$ parity-check matrix. To deal with linear equations we need a map $\nu$ sending a linear polynomial in $\mathbb{F}[x_1,\ldots,x_{3q}]$  to a vector in $\mathbb{F}^{10q}$.  We define $\nu$ as 
\begin{equation}\nonumber
	\nu(f) = a_fM_{\tau}^\top\;,
\end{equation} 
with $M_{\tau}$ as in \eqref{eq def Mtau}.
\begin{Ex}
	Assume $q=2$, then we are working in $\mathbb{F}[x_1,\ldots,x_{6}]$.
	Let $f = x_1 + x_3 + x_5$ and $a_f = (1,0,1,0,1,0)$. Since $q=2$ then $M_{\tau}^\top$ is the matrix
	$$
			M_{\tau}^\top = 
		\left[\begin{array}{cc|cc} 
	\mathbb{I}_3 & \mathbf{0}& \mathbf{0}& \mathbf{0}\\
	\hline 
	\mathbf{0}& \mathbf{0} &\mathbb{I}_3 &\mathbf{0}
\end{array}\right]\;.
	$$
We obtain
	\begin{equation}\nonumber
		\nu(f) = a_fM_{\tau}^\top = (1,0,1,0,0,0,0,0,0,0,0,1,0,0,0,0,0,0,0,0)
	\end{equation}
\end{Ex}
The following lemma will be used to prove the correctness of our reduction.
\begin{Lemma}\label{lem:linvariery}
Let $f = \sum_{i=1}^{3q}a_ix_i + \delta \in \mathbb{F}[x_1,\ldots,x_{3q}]$ be a linear polynomial. Let $\widetilde{v} \in \mathbb{F}^{10q}$.  $\widetilde{v}\cdot\nu(f)^\top =\delta$ if and only if
$
	\tau(\widetilde{v}) \in \mathcal{V}_{\mathbb{F}}(\langle f \rangle).
$
\end{Lemma}
\begin{proof}
	By Remark \ref{rem:invproj}, observing that 
	\begin{equation}\nonumber
		\widetilde{v}\cdot\nu(f)^\top = \widetilde{v}\cdot \left( a_fM_{\tau}^\top\right)^\top =
		\widetilde{v}M_{\tau}a_f^\top=
		\left( \widetilde{v} M_{\tau}\right)\cdot a_f^\top = \tau(\widetilde{v})\cdot a_f^\top = f(\tau(\widetilde{v}))\;.
	\end{equation}
\end{proof}

We construct now the MLD instance $\left( H,s,t \right)$ for a general MQ system. The basic idea is to see the parity-check matrix $\widetilde{H}$ we built so far as a matrix of coefficients for an equation system. Adding rows to the matrix means adding new equations to the system and thus reducing the solution space. 
Consider a standard-form system $S \in \mathbb{F}[x_1,\ldots, x_{3q}]$ containing $q$ quadratic equations and $\lambda$ linear equations, as in definition \ref{def:standardform}.  Let $K = \langle S \rangle$ and $\mathcal{V}_{\mathbb{F}}(K) \subset \mathbb{F}^{10q}$ be its variety. Denote by $S_Q \subset S$ the subset of quadratic equations in $S$, and by $\langle S_Q \rangle$ the ideal generated by it.
 
 Let also $\left( \widetilde{H}, \widetilde{s}, \widetilde{t}\right)$ be the MLD instance corresponding to $S_Q$. Let $\lbrace f_1,\ldots,f_{\lambda} \rbrace \subset S$ be the set of linear equations of $S$.
We build a parity-check matrix $H$ as follows
\begin{equation}\label{eq:final H}
	H = \begin{bmatrix}
		\widetilde{H}\\
		\nu(f_1)\\
		\vdots\\
		\nu(f_\lambda)
	\end{bmatrix} \in \mathbb{F}^{(7q+\lambda) \times 10q}
\end{equation}
Consider the following syndrome vector
\begin{equation}\nonumber
	s = \widetilde{s} \| \delta_1 \| \cdots \| \delta_{\lambda}\;,
\end{equation}
where $f_i(0)=\delta_i$ for any $i$,
and set $t = \widetilde{t} = 3q$.

Consider furthermore the set $W = \lbrace \widetilde{v} \in \widetilde{W} \mid \widetilde{v}\cdot \nu(f_i)^\top = \delta_i \;\forall i=1,\ldots,\lambda  \rbrace\subset \widetilde{W}$. 

\begin{theorem}\label{th:final}
$W$ is the set of witnesses for the MLD instance $\left(H,s,t\right)$. Moreover $\mathcal{V}_{\mathbb{F}}(K) = \tau\left( W \right)$.
\end{theorem}
\begin{proof}
We need to prove that given $\widetilde{v} \in W$ it holds $H\widetilde{v}^\top = s^\top$ and  $\mathrm{w}(\widetilde{v}) \le t$.
By definition of $W$, we obtain that $\widetilde{v}\in W$ implies $\widetilde{v} \in \widetilde{W}$, which means $\mathrm{w}(\widetilde{v}) \le \widetilde{t}$ and also $\widetilde{H}\widetilde{v}^\top = \widetilde{s}^\top$. Moreover we have $\widetilde{v}\cdot\nu(f_i)^\top = \delta_i$ for every $i=1,\ldots,\lambda$ implying 
\begin{equation}\nonumber
	\begin{bmatrix}
		\widetilde{H}\\
		\nu(f_1)\\
		\vdots\\
		\nu(f_{\lambda})
	\end{bmatrix}\widetilde{v}^\top = \begin{pmatrix}\widetilde{s}\\\delta_1\\\vdots\\\delta_{\lambda}\end{pmatrix} = s^\top\;.
\end{equation}

Due to lemma \ref{lem:linvariery} we have $\tau(\widetilde{v}) \in \mathcal{V}_{\mathbb{F}}(\langle f_i\rangle)$ for $i=1,\ldots,\lambda$. Therefore $\tau(\widetilde{v}) \in \bigcap_{i=1}^{\lambda} \mathcal{V}_{\mathbb{F}}(\langle f_i\rangle) = \mathcal{V}_{\mathbb{F}}(\langle f_1,\ldots,f_{\lambda}\rangle)$. Since by Proposition \ref{prop: multiplequadeq} we have $\tau(\widetilde{v}) \in \mathcal{V}_{\mathbb{F}}(S_q)$, then $\tau(\widetilde{v})\in \mathcal{V}_{\mathbb{F}}(\langle S_q \rangle) \cap \mathcal{V}_{\mathbb{F}}(\langle f_1,\ldots,f_{\lambda}\rangle) = \mathcal{V}_{\mathbb{F}}(\langle S \rangle) = \mathcal{V}_{\mathbb{F}}(K)$.

On the other hand let $z \in \mathcal{V}_{\mathbb{F}}(\langle K \rangle) = \mathcal{V}_{\mathbb{F}}(\langle S_q \rangle) \cap \mathcal{V}_{\mathbb{F}}(\langle f_1,\ldots,f_{\lambda}\rangle) = \tau(\widetilde{W})\cap \mathcal{V}_{\mathbb{F}}(\langle f_1,\ldots,f_{\lambda}\rangle) \subseteq \tau(\widetilde{W})$, where the last equality comes from Proposition \ref{prop: multiplequadeq}.
Therefore there exists $z' \in \tau(\widetilde{W})$ such that $z = \tau(z')$. 
Since $a_{f_i}\cdot z^\top = a_{f_i}\cdot \tau(z')^\top = \delta_i$ for every $i=1,\ldots,\lambda$, this implies $\nu(a_{f_i})\cdot z'^\top = \delta_i$ and therefore $z \in \tau(W)$.
\end{proof}

We can now define the map $\beta : \mathrm{MQ} \rightarrow \mathrm{MLD}$ as follows
\begin{equation}\nonumber
\beta(S) = \left( H,s,t\right)\;
\end{equation}
where $H$, $s$ and $t$ are as in Theorem \ref{th:final}. We now prove that a witness of such instance can be transformed into a witness of $S$ by applying the truncation $\tau$.

\begin{theorem}
Let $I_{\mathrm{MQ}} =S$ where $S$ is a standard form system of quadratic equations. If $\widetilde{v} \in \mathbb{F}^{10q}$ is a solution of $\beta(I_{\mathrm{MQ}})$ then $\tau(\widetilde{v})$ is a solution of $I_{\mathrm{MQ}}$.
\end{theorem}
\begin{proof}
Let $K = \langle S \rangle$ and apply theorem \ref{th:final}. If $\widetilde{v}$ solves $\beta(I_{\mathrm{MQ}}) = \left(H,s,t \right)$ then $\tau(\widetilde{v}) \in \mathcal{V}_{\mathbb{F}}(K)$, i.e. $\tau(\widetilde{v})$ is a solution of $S$.
\end{proof}

\begin{theorem}
Given a MQ system $S \in \mathbb{F}[x_1,\ldots,x_{\nmq}]$ consisting of $\mmq$ equations, the reduction $\beta$ runs in polynomial space bounded by $\mathcal{O}(\nmq^4\mmq^2)$. 
\end{theorem}
\begin{proof}
Recall that by Lemma \ref{lem:systostandard} we can transform $S$ into a standard form system $S'$ in $\mathcal{O}(\nmq^2\mmq)$ operations. This process produces $S' = S'_Q \cup S'_L$ where $S'_Q$ and $S'_L$ are sets of quadratic and linear equations, namely. Let $q= \abs{S'_Q}
\leq
\mmq\left(\frac{\nmq(\nmq-1)}{2}\right)\leq \mmq\nmq^2$
and $\lambda= \abs{S'_L}
\leq 
	\mmq\left(\frac{3\nmq^2-\nmq}{2}-2\right)\leq \frac{3}{2}\mmq\nmq^2$.
We estimate only the construction of the matrix $H$ since $s$ and $t$ have a significant smaller size. 
\\
The matrix $H$ generated via $\beta$ has dimension $(7q+\lambda) \times 10q$ as in \eqref{eq:final H}, therefore it takes space at most $$(7(\mmq\nmq^2) + \frac{3}{2}\mmq\nmq^2) \cdot 10\mmq\nmq^2=
85\mmq^2\nmq^4
\in \mathcal{O}(\mmq^2\nmq^4)$$.
\end{proof}

\section{Conclusions}\label{sec: conclusion}
In this work we introduced two polynomial-time reductions: $\alpha$, from MLD to MQ, and $\beta$, from MQ to MLD. Therefore, the composition of $\alpha$ and $\beta$ is a polynomial-time auto-reduction of $\mathrm{MLD}$, while $\alpha\circ\beta$ is a polynomial-time auto-reduction of $\mathrm{MQ}$.
Hence each MLD instance can be solved if we are able to solve each MLD instance defined by $\beta\circ \alpha$ and even more so if we are able to solve those defined by $\beta$. Similarly, each MQ instance can be solved if we are able to solve each MQ instance defined by $\alpha\circ \beta$, or even only defined by $\alpha$.
So if we can decide in polynomial time the existence of solutions for all systems in the image of $\alpha$, then we can solve MQ in polynomial time.
Notice that the same property holds for systems in standard form.
So we can identify two families of systems which play a special role in the MQ problem: one is classical and the other comes from our reduction.
\\
In the case of MLD there exist families of codes that plays a similar role in the MLD context, for example the one obtained via the reductions in \cite{Berlekamp} and the one obtained with our results.
We can formalise this in the following theorem.

\begin{theorem}\label{thm standard mld}
Let $\mathcal{C}$ be the family of codes defined by parity-check matrices as in \eqref{eq:final H}. \\
If we can solve all MLD instances for $\mathcal{C}$ in polynomial time, then we can solve in polynomial time all instances of MLD, and so, P=NP.
\end{theorem}

Regarding the relation between MLD and MQ, we remark that two NP-complete problems might be not isomorphic (see \cite{berman1977isomorphisms} for a formal definition of polynomial-time isomorphism and for the Berman-Hartmanis conjecture on isomorphic NP problems). 
We rephrase here \cite[~Th.1]{berman1977isomorphisms} since it is needed in our subsequent discussion. 
\begin{theorem}[{\cite[~Th.1]{berman1977isomorphisms}}]\label{thm original isomorphism}
If there are two length-increasing invertible p-reductions, one of $A$ to $B$ and the other of $B$ to $A$, then $A$ and $B$ are isomorphic.
\end{theorem}
In our case $A$ and $B$ are MLD and MQ, while the two reductions are $\alpha$ and $\beta$, which are polynomial-time length-increasing reductions. However, only $\alpha$ can be inverted, since $\beta$ requires the reduction of MQ instances into standard-form systems (a many-to-one reduction), hence the hypotheses of Theorem~\ref{thm original isomorphism} are not completely satisfied. To obtain a one-to-one reduction from MQ to MLD, we can modify the definition of standard-form systems by adding new equations containing the information about the original MQ instance. For example, we can consider an additional set of $\mmq\left(\binom{\nmq}{2}+\nmq+1\right)$ equations and new variables of the form
$$
\left\{
\begin{array}{ll}
\bar{\gamma}_{ij}^{(h)}+\gamma_{ij}^{(h)}=0& \qquad 1\leq h\leq \mmq, 1\leq i<j\leq \nmq\\
\bar{\lambda}_{i}^{(h)}+\lambda_{i}^{(h)}=0& \qquad 1\leq h\leq \mmq,1\leq i\leq \nmq\\
\bar{\delta}^{(h)}+\delta^{(h)}=0&\qquad 1\leq h\leq \mmq\\
\end{array}
\right.
$$
namely, these equations specify the monomials' coefficients of the polynomials in the original MQ instance (see \eqref{eq: MQ instance}).
\\
Once we modify the definition of standard-form instance (and $\beta$ accordingly), we obtain a one-to-one reduction from MQ to MLD. In this way both $\alpha$ and $\beta$ satisfy the hypotheses of Theorem \ref{thm original isomorphism}, thus proving the existence of an isomorphism between MLD and MQ.
\begin{theorem}
MLD and MQ are isomorphic.
\end{theorem}
This isomorphism shows our claimed equivalence and it implies the importance of studying the security of code-based and multivariate-based schemes by meaning of both methods from Coding Theory and Computational Algebra.

\section{Open problems}\label{sec: open problems}
We highlight here a few directions for future works. 
\begin{itemize}
    \item An investigation of the image of $\alpha$. Polynomial systems obtained via $\alpha$ have a special form, and by analysing them new hints on the difficulty of MLD could be determined. Similarly for polynomials coming from $\alpha\circ \beta$. In particular, MLD instances obtained from code-based cryptosystems are of particular interests. We recall that, in terms of MLD, most cryptographic code-based schemes are modeled as triples $(H,s,t)$ with $H$ defined as a permuted version of the parity-check matrix of an algebraic code (for instance, in Classic McEliece \cite{bernstein2017classic}, $H$ hides the parity-check matrix of a binary irreducible Goppa code). Loosely speaking, only who can solve the instance can decrypt a ciphertext, and, to the current knowledge, this is feasible only to those who know the hidden Goppa code. However, some vulnerabilities may be revealed by looking at the associated MQ instance. 
    \item Analogously, for codes coming from the image of $\beta$ and $\beta\circ\alpha$. In particular, similar to above, it would be interesting to focus on MQ instances corresponding to multivariate-based cryptosystems. For instance, in the Digital Signature Scheme Rainbow \cite{Rainbow} the public key is a masked version of a quadratic Boolean polynomial system for which there exists a fast solving algorithm based on Gaussian elimination. It appears that only those who know the original form of the polynomial system are capable of signing messages. However, it may be the case that hidden vulnerabilities are disclosed by applying $\beta$ to these instances.
    \item Even though apparently similar to the directions hinted above, the third future work we propose is even more linked to the security of code-based and multivariate-based cryptosystems. As already introduced, both post-quantum cryptographic families rely on the security of two kinds of problems, the first is the NP-hard problem of solving a generic instance (i.e. MLD for code-based ciphers and MQ for multivariate-based ciphers), the second is the ability of distinguishing between a generic instance and the masked easy-to-solve underlying algebraic instance (e.g. the Goppa code hidden inside a Classic McEliece public key or the multi-level Oil\&Vinegar system hidden inside a Rainbow public key). Instead of blindly applying our reductions to cryptographic public keys, it would be important to model the precise problem of extrapolating the private keys from the public keys, and thus study the complexity of attacking code-based and multivariate-based schemes by understanding their key-generation algorithms.
\end{itemize}




\section*{Acknowledgments}
The publication was created with the co-financing of the European Union -  FSE-REACT-EU, PON Research and Innovation 2014-2020 DM1062 / 2021.
The core of this work is contained in the second author's MSC thesis, supervised by the other authors. Part of this work was presented in 2019 by the first author at the Dipartimento di Ingegneria e Scienze dell'Informazione e Matematica of the University of L'Aquila, Italy.
The first author is a  member of the INdAM Research group GNSAGA.
\\
The authors would like to thank Stefano Baratella and Roberto Sebastiani for their helpful comments.

\renewcommand{\bibname}{References} 
\bibliographystyle{plain}
\bibliography{main}

\begin{thebibliography}{10}

\bibitem{avanzi2017crystals}
Roberto Avanzi, Joppe Bos, L{\'e}o Ducas, Eike Kiltz, Tancr{\`e}de Lepoint,
  Vadim Lyubashevsky, John~M Schanck, Peter Schwabe, Gregor Seiler, and Damien
  Stehl{\'e}.
\newblock {CRYSTALS}-{K}yber algorithm specifications and supporting
  documentation.
\newblock {\em Submission to the NIST's post-quantum cryptography
  standardization process}, 2017.

\bibitem{ducas2021crystals}
Shi Bai, Léo Ducas, Eike Kiltz, Tancrède Lepoint, Vadim Lyubashevsky, Peter
  Schwabe, Gregor Seiler, and Damien Stehlé.
\newblock {CRYSTALS}-{D}ilithium - algorithm specifications and supporting
  documentation (version 3.1), Feb. 2021.

\bibitem{bassoSABER}
Andrea Basso, Jose~Maria Bermudo~Mera, Jan-Pieter D'Anvers, Angshuman Karmakar,
  Sujoy~Sinha Roy, Michiel Van~Beirendonck, and Frederik Vercauteren.
\newblock {SABER}: Mod-{LWR} based {KEM} ({R}ound 3 submission).
\newblock {\em Submission to the NIST's post-quantum cryptography
  standardization process}, 2020.

\bibitem{Berlekamp}
Elwyn Berlekamp, Robert McEliece, and Henk Van~Tilborg.
\newblock On the inherent intractability of certain coding problems (corresp.).
\newblock {\em IEEE Transactions on Information Theory}, 24(3):384--386, 1978.

\bibitem{berman1977isomorphisms}
Leonard Berman and Juris Hartmanis.
\newblock On isomorphisms and density of {NP} and other complete sets.
\newblock {\em SIAM Journal on Computing}, 6(2):305--322, 1977.

\bibitem{bernstein2009}
Daniel~J Bernstein, Johannes Buchmann, and Erik Dahmen.
\newblock Post-quantum cryptography, 2009.

\bibitem{bernstein2017classic}
Daniel~J Bernstein, Tung Chou, Tanja Lange, Ingo von Maurich, Rafael Misoczki,
  Ruben Niederhagen, Edoardo Persichetti, Christiane Peters, Peter Schwabe,
  Nicolas Sendrier, et~al.
\newblock Classic {M}c{E}liece: conservative code-based cryptography.
\newblock {\em Submission to the NIST's post-quantum cryptography
  standardization process}, 2017.

\bibitem{Bruck}
Jehoshua Bruck and Moni Naor.
\newblock The hardness of decoding linear codes with preprocessing.
\newblock {\em IEEE Transactions on Information Theory}, 36(2):381--385, 1990.

\bibitem{bucero}
Marta~Abril Bucero and Bernard Mourrain.
\newblock Border basis relaxation for polynomial optimization.
\newblock {\em Journal of Symbolic Computation}, 74:378--399, 2016.

\bibitem{GeMSS}
Antoine Casanova, Jean-Charles Faugere, Gilles Macario-Rat, Jacques Patarin,
  Ludovic Perret, and Jocelyn Ryckeghem.
\newblock {GeMSS}: a great multivariate short signature.
\newblock {\em Submission to the NIST's post-quantum cryptography
  standardization process}, 2017.

\bibitem{chen2019NTRU}
Cong Chen, Oussama Danba, Jeffrey Hoffstein, Andreas Hülsing, Joost Rijneveld,
  John~M. Schanck, Peter Schwabe, William Whyte, and Zhenfei Zhang.
\newblock {NTRU} algorithm specications and supporting documentation.
\newblock {\em Submission to the NIST's post-quantum cryptography
  standardization process}, 2019.

\bibitem{Cheng}
Qi~Cheng.
\newblock Hard problems of algebraic geometry codes.
\newblock {\em IEEE Transactions on Information Theory}, 54(1):402--406, 2008.

\bibitem{cox}
DA~Cox, J~Little, and D~O'Shea.
\newblock Using algebraic geometry, ser.
\newblock {\em Graduate Texts in Mathematics. Springer}, 2005.

\bibitem{davenport}
Mark~A Davenport and Justin Romberg.
\newblock An overview of low-rank matrix recovery from incomplete observations.
\newblock {\em IEEE Journal of Selected Topics in Signal Processing},
  10(4):608--622, 2016.

\bibitem{dennis}
John~E Dennis~Jr and Robert~B Schnabel.
\newblock Numerical methods for unconstrained optimization and nonlinear
  equations, 1996.

\bibitem{Rainbow}
Jintai Ding, Ming-Shing Chen, Albrecht Petzoldt, Dieter Schmidt, Bo-Yin Yang,
  Matthias Kannwischer, and Jacques Patarin.
\newblock Rainbow - algorithm specification and documentation.
\newblock {\em Submission to the NIST's post-quantum cryptography
  standardization process}, 2015.

\bibitem{fouque2018falcon}
Pierre-Alain Fouque, Jeffrey Hoffstein, Paul Kirchner, Vadim Lyubashevsky,
  Thomas Pornin, Thomas Prest, Thomas Ricosset, Gregor Seiler, William Whyte,
  and Zhenfei Zhang.
\newblock {FALCON}: {Fast-Fourier} lattice-based compact signatures over
  {NTRU}.
\newblock {\em Submission to the NIST's post-quantum cryptography
  standardization process}, 2020.

\bibitem{Gandikota}
Venkata Gandikota, Badih Ghazi, and Elena Grigorescu.
\newblock On the {NP}-hardness of bounded distance decoding of {R}eed-{S}olomon
  codes.
\newblock In {\em 2015 IEEE International Symposium on Information Theory
  (ISIT)}, pages 2904--2908. IEEE, 2015.

\bibitem{garey}
Michael~R Garey and David~S Johnson.
\newblock Computers and intractability, 1979.

\bibitem{grigoriev}
Dima Grigoriev and Dmitrii~V Pasechnik.
\newblock Polynomial-time computing over quadratic maps i: sampling in real
  algebraic sets.
\newblock {\em Computational complexity}, 14(1):20--52, 2005.

\bibitem{Guruswami}
Venkatesan Guruswami and Alexander Vardy.
\newblock Maximum-{L}ikelihood {D}ecoding of {R}eed-{S}olomon codes is
  {NP}-hard.
\newblock {\em IEEE Transactions on Information Theory}, 51(7):2249--2256,
  2005.

\bibitem{Karp}
Richard~M Karp.
\newblock Reducibility among combinatorial problems, 1972.

\bibitem{kpg99}
Aviad Kipnis, Jacques Patarin, and Louis Goubin.
\newblock Unbalanced {O}il and {V}inegar signature schemes.
\newblock In {\em International Conference on the Theory and Applications of
  Cryptographic Techniques}, pages 206--222. Springer, 1999.

\bibitem{lasserre}
Jean~B Lasserre.
\newblock Global optimization with polynomials and the problem of moments.
\newblock {\em SIAM Journal on optimization}, 11(3):796--817, 2001.

\bibitem{li}
Tien-Yien Li.
\newblock Numerical solution of multivariate polynomial systems by homotopy
  continuation methods.
\newblock {\em Acta numerica}, 6:399--436, 1997.

\bibitem{Lobstein}
Antoine Lobstein.
\newblock The hardness of solving subset sum with preprocessing.
\newblock {\em IEEE Transactions on Information Theory}, 36(4):943--946, 1990.

\bibitem{mi88}
Tsutomu Matsumoto and Hideki Imai.
\newblock Public quadratic polynomial-tuples for efficient
  signature-verification and message-encryption.
\newblock In {\em Workshop on the Theory and Application of of Cryptographic
  Techniques}, pages 419--453. Springer, 1988.

\bibitem{mceliece1978public}
Robert~J McEliece.
\newblock A public-key cryptosystem based on algebraic coding theory.
\newblock {\em Coding Thv}, 4244:114--116, 1978.

\bibitem{menezes2018handbook}
Alfred~J Menezes, Paul~C Van~Oorschot, and Scott~A Vanstone.
\newblock Handbook of applied cryptography, 2018.

\bibitem{micciancio2002complexity}
Daniele Micciancio and Shafi Goldwasser.
\newblock Complexity of lattice problems: a cryptographic perspective, 2002.

\bibitem{niederreiter1986knapsack}
Harald Niederreiter.
\newblock Knapsack-type cryptosystems and algebraic coding theory.
\newblock {\em Prob. Contr. Inform. Theory}, 15(2):157--166, 1986.

\bibitem{parrillo}
Pablo~A Parrilo.
\newblock Semidefinite programming relaxations for semialgebraic problems.
\newblock {\em Mathematical programming}, 96(2):293--320, 2003.

\bibitem{pat96}
Jacques Patarin.
\newblock Hidden fields equations (hfe) and isomorphisms of polynomials (ip):
  Two new families of asymmetric algorithms.
\newblock In {\em International Conference on the Theory and Applications of
  Cryptographic Techniques}, pages 33--48. Springer, 1996.

\bibitem{Peterson}
Wesley Peterson.
\newblock Encoding and error-correction procedures for the {B}ose-{C}haudhuri
  codes.
\newblock {\em IRE Transactions on information theory}, 6(4):459--470, 1960.

\bibitem{schnabel}
Robert~B Schnabel and Paul~D Frank.
\newblock Tensor methods for nonlinear equations.
\newblock {\em SIAM Journal on Numerical Analysis}, 21(5):815--843, 1984.

\bibitem{kipnis}
Adi Shamir and Aviad Kipnis.
\newblock Cryptanalysis of the {HFE} public key cryptosystem.
\newblock In {\em Advances in Cryptology, Proceedings of Crypto}, volume~99,
  1999.

\bibitem{shor1994algorithms}
Peter~W Shor.
\newblock Algorithms for quantum computation: discrete logarithms and
  factoring.
\newblock In {\em Proceedings 35th annual symposium on foundations of computer
  science}, pages 124--134. Ieee, 1994.

\bibitem{stinson2019cryptography}
DR~Stinson and MB~Paterson.
\newblock Cryptography theory and practice (fourth edi), 2019.

\bibitem{Vardy}
Alexander Vardy.
\newblock The intractability of computing the minimum distance of a code.
\newblock {\em IEEE Transactions on Information Theory}, 43(6):1757--1766,
  1997.

\bibitem{verschelde}
Jan Verschelde.
\newblock Polynomial homotopies for dense, sparse and determinantal systems.
\newblock {\em arXiv preprint math/9907060}, 1999.

\end{thebibliography}



\end{document}